\documentclass{article}
\usepackage{arxiv}
\usepackage{pifont}
\usepackage{amsmath,amsthm,amssymb,comment,microtype,graphicx}
\usepackage{multicol}
\usepackage{amstext}
\usepackage{multirow}
\usepackage{booktabs}
\usepackage[skip=0pt]{subcaption}
\usepackage{times}
\usepackage{lipsum}
\usepackage[shortlabels]{enumitem}
\usepackage{cancel}
\usepackage{wrapfig}
\usepackage{array}
\usepackage{siunitx}
\usepackage{csvsimple}
\usepackage[multidot]{grffile}
\usepackage{bbm}
\usepackage{psfrag} 
\usepackage{diagbox}
\usepackage{bm}
\usepackage{dblfloatfix}
\usepackage{makecell}
\usepackage{dsfont}
\usepackage{mathtools}
\usepackage{xcolor}
\usepackage[multiple]{footmisc}
\usepackage{mathrsfs}
\usepackage{todonotes}
\usepackage{setspace} 
\usepackage{tikz}
\usepackage{nicefrac}
\usepackage{algorithm}
\usepackage{algpseudocode}
\usepackage{bbold}
\usepackage{xspace, soul}
\usepackage{braket}
\usepackage{suffix}

\newtheorem{theorem}{Theorem}[section]
\newtheorem{lemma}{Lemma}[section]

\newtheorem{definition}{Definition}[section]

\newtheorem{remark}{Remark}[section]

\usepackage{pgfplots}
\usepackage{pgfplotstable, booktabs}
\usepackage{filecontents}
\usepgfplotslibrary{groupplots}
\usetikzlibrary{patterns}
\pgfkeys{/pgf/number format/.cd,1000 sep={\,}}
\tikzset{every picture/.style={line width=0.75pt}} 
\pgfplotsset{
	ylabel right/.style={
		after end axis/.append code={
			\node [rotate=90, anchor=north] at (rel axis cs:1,0.5) {#1};
		}
	}
}
\usetikzlibrary{arrows,positioning}
\usepackage{hyperref}

\usepackage[utf8]{inputenc} 
\usepackage[T1]{fontenc}    
\usepackage{hyperref}       
\usepackage{url}            
\usepackage{booktabs}       
\usepackage{amsfonts}       
\usepackage{nicefrac}       
\usepackage{microtype}      
\usepackage{lipsum}		
\usepackage{graphicx}
\usepackage{natbib}
\usepackage{doi}

\title{Differential Privacy\\ Overview and Fundamental Techniques}

\author{ 
    Ferdinando Fioretto\\
	University of Virginia\\
    \texttt{fioretto@virginia.edu} \\
	\And
    Pascal Van Hentenryck\\
    Georgia Institute of Technology\\
    \texttt{pvh@gatech.edu} \\
    \And
    Juba Ziani\\
    Georgia Institute of Technology\\
    \texttt{jziani3@gatech.edu} \\
}

\date{}

\renewcommand{\headeright}{Book Chapter}

\hypersetup{
pdftitle={A template for the arxiv style},
pdfsubject={q-bio.NC, q-bio.QM},
pdfauthor={David S.~Hippocampus, Elias D.~Striatum},
pdfkeywords={First keyword, Second keyword, More},
}

\begin{document}
\maketitle

\begin{abstract}
This chapter is meant to be part of the book ``Differential Privacy in Artificial Intelligence: From Theory to Practice'' and provides an introduction to Differential Privacy. It starts by illustrating various attempts to protect data privacy, emphasizing where and why they failed, and providing the key desiderata of a robust privacy definition. It then defines the key actors, tasks, and scopes that make up the domain of privacy-preserving data analysis. 
Following that, it formalizes the definition of Differential Privacy and its inherent properties, including composition, post-processing immunity, and group privacy. The chapter also reviews the basic techniques and mechanisms commonly used to implement Differential Privacy in its pure and approximate forms.
\end{abstract}


\section{Introduction}
\label{ch1:sec1}

Data is continuously harvested from nearly every facet of our lives by corporations, service providers, and public institutions. Whether through smartphones, social media interactions, internet use, healthcare visits, or financial transactions, information is continuously gathered, shaping the foundation of the modern economy. Private companies leverage this vast pool of data to evaluate loan candidates, optimize transportation networks, improve supply chains, personalize services, and predict market demands, all to enhance decision making. Similarly, public policies and government initiatives rely heavily on this data, guiding resource distribution, monitoring public health crises, and driving urban development and sustainability efforts.

However, these datasets also contain a large array of sensitive information, including health, financial, or location data. Major privacy violations and breaches are commonplace, and can have severe negative impacts, not only on  consumers and online users, but also on entire organizations and governments. 
For instance, the 2017 Equifax data breach~\citep{equifax_wiki} exposed the personal information of 147 million individuals, including social security numbers, birth dates, and addresses, leaving millions vulnerable to identity theft, fraud, and long-term financial harm. Similarly, the 2016 Facebook-Cambridge Analytica scandal~\citep{facebook_analytica_wiki}, in which the personal data of up to 87 million Facebook users was harvested without consent for political advertising purposes, raised concerns about its possible influence on the outcome of the 2016 presidential election.

Privacy concerns have become central in today’s society, driving significant changes in government policy. Various regulatory frameworks have been established, with the United States and Europe leading efforts toward stronger privacy practices.  In Europe, the General Data Protection Regulation (GDPR) sets strict standards for data management, focusing on consent and data minimization \citep{GDPR}, while in the U.S., regulations like Title 13 \citep{CongressTitle13} govern the handling of census data and laws like the Health Insurance Portability and Accountability Act (HIPAA) and the California Consumer Privacy Act (CCPA) offer protections for health and consumer data \citep{hipaa, CCPA}. 
This movement was further emphasized in October 2023 when the Biden administration issued an Executive Order on AI, ensuring the enforcement of consumer protection laws and introducing safeguards against privacy violations in AI systems. Government actions, such as the release of the AI Bill of Rights Blueprint~\citep{billofrights} in the US, underscore the increasing focus on privacy in both policy and technology.

Public policy has also devoted extensive research into technical solutions for privacy. Over the past three decades, this research has explored a wide range of privacy definitions and techniques, but one has emerged as a pivotal framework: \emph{Differential Privacy} (DP)~\citep{DKMMN:06}. DP has gained widespread recognition and adoption, not only by leading technology companies like Apple, Meta, Google, and LinkedIn, but also by the U.S.~government, most notably in its landmark 2020 Census data release.

Differential Privacy is now widely regarded as the gold standard for privacy protection in statistical analyses and dataset releases. Its strength lies in providing a \emph{formal} and \emph{mathematical} definition of privacy, offering \emph{precise} and \emph{provable} guarantees. This is in stark contrast to historically ad-hoc and loosely defined privacy methods, which have repeatedly failed under attacks aimed at reconstructing part of the original dataset or identifying individuals in said datasets. As privacy challenges evolve, so too does Differential Privacy, expanding across diverse fields to meet new demands. 
This book aims at providing a comprehensive introduction to DP, particularly within the novel challenges brought by AI applications. It explores its foundational theories, applications in machine learning, and practical implementations, equipping readers with the knowledge to leverage this critical technology effectively.

\paragraph{Overview of the chapter.} 
This chapter is structured to provide an introduction to Differential Privacy. It begins by illustrating various attempts to protect data privacy, emphasizing where and why they failed, and providing the key desiderata of a robust privacy definition (Section~\ref{ch1:sec2}).
It then defines the key actors, tasks, and scopes that make up the domain of privacy-preserving data analysis (Section~\ref{ch1:sec3}).
Following that, Section~\ref{ch1:sec4}, formalizes the definition of DP and its inherent properties, including composition, post-processing immunity, and group privacy. The chapter also reviews the basic techniques and mechanisms commonly used to implement Differential Privacy in Sections \ref{ch1:sec4} to \ref{ch1:sec6}. Finally, Section \ref{ch1:sec7} concludes with an overview of Differential Privacy applications and some future directions in this field.

\section{A Historical Perspective on Privacy}
\label{ch1:sec2}

This section begins by posing a fundamental question: \emph{What are the key desiderata and properties that a robust privacy definition must guarantee?} To address this, it examines historical failures of previous and current privacy definitions, highlighting the necessity for well-defined and formal guarantees. This section first outlines the main properties satisfied by Differential Privacy---these properties will be formally detailed later in this chapter. It then delves into specific examples of major privacy breaches over the past 30 years, identifying for each how adherence to certain privacy desiderata could have prevented the failure.

A central argument of this book is the importance of \emph{well-defined and formal privacy guarantees}. A major weakness in many privacy techniques arises when the protections themselves are poorly specified, particularly when they fail to clearly define the classes of attacks they are designed to resist. Over the past three decades, numerous privacy attacks have exploited such ambiguities, often by applying privacy notions beyond their intended use cases. To address these challenges, this chapter focuses on four main \emph{desiderata} that a strong privacy definition should satisfy:

\begin{enumerate}[leftmargin=*,topsep=2pt,itemsep=2pt]
\item \textbf{Desiderata 1: Compositionality.} A good privacy definition should ensure that its protections gracefully degrade  when applied multiple times, whether across several datasets or through repeated private data analyses. In a data-driven world, where datasets are frequently analyzed multiple times and may contain overlapping information about individuals, composition is crucial. Without it, repeated analyses can cumulatively erode privacy safeguards and ultimately compromise individual privacy.

\item \textbf{Desiderata 2: Post-processing immunity.} 
Once data has been privatized using a privacy-preserving mechanism, any further data analyses should not degrade its privacy guarantees, provided that the original, non-privatized data remains inaccessible. This property assures that subsequent steps or transformations applied to the privatized output cannot compromise privacy. Post-processing immunity offers a strong guarantee that allows data analysts to abstract away potential attack models, effectively providing \emph{future-proof protection} against privacy violations.

\item \textbf{Desiderata 3: Group privacy.} 
Group privacy aims at controlling how privacy guarantees degrade when considering groups of individuals rather than single individuals. It ensures that a privacy mechanism does not arbitrarily fail to protect privacy beyond the individual level when data from multiple users is combined. While it is inevitable that privacy guarantees weaken as group sizes increase, since more information is encoded about them, the degradation should be controlled and quantifiable.

\item \textbf{Desiderata 4: Quantifiable privacy-accuracy trade-offs.} 
There is no free lunch in privacy: releasing accurate information about a group of people must necessarily and statistically encode some information about individuals. As privacy protection increases the accuracy of insights derived from the data may decrease. A good privacy definition should provide quantifiable trade-offs, allowing data analysts, decision-makers, and model builders to measure how much accuracy is sacrificed for a given level of privacy. This enables them to balance privacy and utility according to specific needs. 
\end{enumerate}

\noindent
The following sections provide historical examples illustrating why privacy is complex, where traditional methods have failed, and how the above desiderata are essential for guaranteeing robust privacy.

\subsection{Data Anonymization} 

A standard technique for privacy protection in various domains is \emph{anonymization}. It involves the removal or masking of any identifying details to prevent the recovery of personal identities. Anonymization has been employed in areas such as the release of medical datasets under the Health Insurance Portability and Accountability Act (HIPAA) standards. In the mid-1990s, the Massachusetts Group Insurance Commission (GIC), a government agency responsible for purchasing health insurance for state employees, sought to promote medical research by releasing anonymized health data. The GIC approach involved removing what they considered ``explicit'' identifiers such as names, addresses, and social security numbers, while retaining hundreds of other attributes deemed non-identifiable. Supported by then-Governor William Weld, this initiative aimed at balancing data utility with privacy protection. However, in 1997, Dr.~Latanya Sweeney, then a graduate student at MIT, set out to challenge the effectiveness of this anonymization. Using publicly available information, she re-identified Governor Weld’s medical records within the dataset and sent them to his office, starkly demonstrating the vulnerability of supposedly anonymized data.

How was Dr. Sweeney able to uncover Governor Weld’s personal medical information from the GIC’s released data? One might assume that such an attack required sophisticated techniques and significant resources. In reality, her de-anonymization attack cost only \$20 and limited time. The GIC’s dataset included three crucial attributes for each individual: sex, zip code, and date of birth. Dr.~Sweeney purchased voter registration records from Cambridge, Massachusetts, which contained names, addresses, zip codes, and dates of birth. By cross-referencing these two datasets, she found that only six people in Cambridge shared Governor Weld’s birth date. Of those, only three were male, and just one resided in his zip code—uniquely identifying his medical records. This type of attack, known as a \emph{linkage attack}, re-identifies individuals by linking anonymized data with external public records. In a subsequent report~\citep{sweeney2000simple}, Dr. Sweeney demonstrated that her attack extended far beyond a single high-profile individual. She found that ``87\% of the population in the United States had reported characteristics that likely made them unique based only on {5-digit ZIP, gender, date of birth}.'' Even at broader geographic levels, significant portions of the population could be uniquely identified with minimal information.
``About half of the U.S.~population 
are likely to be uniquely identified by only {place, gender, date of birth}, where place indicates the city, town, or municipality in which the individual resides.'' 

\paragraph{Why did anonymization fail?} Anonymization failed because it lacked formal privacy guarantees. Dr.~Sweeney’s attack was remarkably simple, yet unanticipated due to the absence of a precise attack model. The lack of \emph{post-processing immunity} meant that, once the anonymized data was released, combining it with other publicly available datasets could reveal more information than intended. If the privacy mechanism had been robust to post-processing, additional analyses or data combinations would not have compromised individual privacy beyond what was already publicly accessible. This example motivates the need for formal privacy definitions that account for all potential avenues of data exploitation.

\subsection{K-Anonymity}

At this point, one might argue that the previous example does not represent a fundamental failure of anonymization as a privacy technique, but rather a misapplication in that specific instance. Is it possible to thwart de-anonymization attacks by simply withholding more attributes? For instance, would not releasing someone's zip code, date of birth, or gender resolve the issue? However, a significant challenge emerges in determining which combinations of publicly available attributes could uniquely identify an individual. As the number of features in a dataset grows, it becomes practically impossible for modern computing to predict and guard against all potential attack vectors. Moreover, sensitive attributes often correlate statistically with non-sensitive ones, rendering anonymization susceptible to statistical attacks that can probabilistically reconstruct sensitive information through these correlations—for a particularly sensitive example involving genomic data, refer to~\cite{homer2008resolving}.

Despite decades of deployment, anonymization has consistently failed to provide robust privacy protection. Other high-profile failures include the AOL search data release~\citep{barbaro2006face}, the Netflix Prize dataset~\citep{Narayanan2006HowTB}, and studies demonstrating that individuals can be uniquely identified using just a few mobile phone location points~\citep{de2013unique}. So, what is the next step? Can  the concept of anonymization be refined to address its shortcomings? A promising strategy might be to release only partial information about each attribute. For example, in demographic or medical analyses, knowing that an individual falls within a certain age \emph{range}, such as ``between 18 and 35,'' might suffice. By revealing less precise information, can re-identification attacks be made more difficult?

\begin{table}[!t]
\centering
\begin{minipage}{0.29\textwidth}
\centering
\resizebox{\linewidth}{!}{
\begin{tabular}{c|c|c}
\multicolumn{3}{c}{Original Data} \\ 
\toprule
Name & Zip Code & Age \\ 
\midrule
Rick & 19456 & 67 \\  
Nathan & 30309 & 33 \\  
Yani & 19445 & 64 \\  
Xiao & 30457 & 35 \\  
Luciana & 19456 & 67 \\  
Anastasia & 30271 & 38 \\  
Marcia & 19456 & 31 \\  
Yuki & 19456 & 62 \\  
\bottomrule
\end{tabular}}
\end{minipage}%
\hspace{15pt} 
\begin{minipage}{0.23\textwidth}
\centering
\resizebox{\linewidth}{!}{
\begin{tabular}{c|c|c}
\multicolumn{3}{c}{Anonymized Data} \\ 
\toprule
ID & Zip Code & Age \\ 
\midrule
1 & 19456 & 67 \\  
2 & 30309 & 33 \\  
3 & 19445 & 64 \\  
4 & 30457 & 35 \\  
5 & 19456 & 67 \\  
6 & 30271 & 38 \\  
7 & 19456 & 31 \\  
8 & 19456 & 62 \\  
\bottomrule
\end{tabular}}
\end{minipage}%
\hspace{15pt} 
\begin{minipage}{0.24\textwidth}
\centering
\resizebox{\linewidth}{!}{
\begin{tabular}{c|c|c}
\multicolumn{3}{c}{$4$-anonymized Data} \\ 
\toprule
ID & Zip Code & Age \\ 
\midrule
1 & 19*** & 60-70 \\ 
2 & 30*** & 30-40 \\ 
3 & 19*** & 60-70 \\ 
4 & 30*** & 30-40 \\ 
5 & 19*** & 60-70 \\ 
6 & 30*** & 30-40 \\ 
7 & 30*** & 30-40 \\ 
8 & 19*** & 60-70 \\ 
\bottomrule
\end{tabular}}
\end{minipage}
\caption{Three levels of anonymization on a demographic dataset. 
{\bf Left:} original dataset. 
{\bf Center:} masking the sensitive names for1-anonymity.
{\bf Right:} generalizing Zip codes and age attributes for 4-anonymity.}\label{fig:kanonym} 
\end{table}

In 1998, Prof.~Pierangela Samarati and Dr.~Latanya Sweeney (the same Dr.Sweeney who highlighted the failures of basic anonymization) introduced a generalization called \emph{$k$-anonymity}~\citep{1998kanon,SweeneyKAnon}. A dataset satisfies $k$-anonymity if, for every record, there are at least $k - 1$ other records with identical values in a set of quasi-identifiers—attributes that could potentially be linked to external data to re-identify individuals. In this framework, it should be impossible to distinguish between any of the $k$ individuals sharing the same quasi-identifiers. In particular, the larger the value of $k$, the stronger the privacy guarantee. Consider, for example, the dataset in Table~\ref{fig:kanonym} (left), containing information about state employees. One approach is to release this dataset by replacing sensitive names with random identifiers (see Table~\ref{fig:kanonym} (middle)). This technique provides only $1$-anonymity, which is essentially standard anonymization. However, each individual still has a unique combination of zip code and age, making them vulnerable to singling-out attacks~\citep{sweeney2000simple,SweeneyKAnon}.
In contrast, the table on the right demonstrates $4$-anonymity: entries \#1, 3, 5, and 8 are indistinguishable from each other, as are entries \#2, 4, 6, and 7. Individuals are grouped into clusters of four, where each group shares the same generalized (zip code, age) attributes, significantly enhancing privacy.

\begin{remark}[Privacy~vs.~Utility]
In $k$-anonymization, increasing the value of $k$ enhances \emph{privacy} by making it more difficult to distinguish between individuals, as they are grouped into larger clusters with identical quasi-identifiers. However, this comes at the expense of \emph{utility}. As $k$ increases, the information becomes less precise, reducing the dataset’s usefulness for analysis. For instance, in the $4$-anonymized version of our dataset, the details about individuals’ zip codes and ages are less specific compared to the $1$-anonymized version. This trade-off between privacy and utility is a central theme in privacy research and will be addressed in the context of Differential Privacy in subsequent chapters.
\end{remark}

\paragraph{Where does $k$-anonymization fail? Reason \#1: Lack of group privacy.}
At first glance, $k$-anonymity appears to address the shortcomings of basic anonymization by preventing the singling out of any specific individual within a dataset. In fact, for years, it was considered the state-of-the-art solution for preventing re-identification attacks. However, $k$-anonymity suffers from a significant limitation concerning the leakage of sensitive information, even when individuals are not directly identified.
The core issue is not merely the potential to link a data subject to a specific record. Instead, the real problem lies in the exposure of sensitive information associated with individuals without explicit re-identification, as highlighted in~\citep{desfontainesblog20170814}. In essence, one does not need to pinpoint a specific person to infer personal, sensitive details about them. Consider the previous example, but now suppose the dataset includes a sensitive attribute, such as credit scores (see Table \ref{tab:credit}, Left). Even without directly identifying anyone, an adversary could learn sensitive information about individuals based on the available data.
Using the same linkage approach that Dr.~Sweeney employed in her de-anonymization of the GIC medical records, one could cross-reference publicly available data to deduce that entries \#2, 4, 6, and 7 correspond to Nathan, Xiao, Anastasia, and Marcia. Although it is impossible to match each person to their exact record, one can still infer that all four individuals have a credit score in the  ``Fair'' category. This represents a significant privacy breach, as sensitive information is disclosed without explicit identification. Here, the property of \emph{group privacy} is violated—the privacy guarantee collapses when aggregating data from as few as four individuals—leading to both group-level and \emph{individual-level} harms.

\begin{table}[!t]
\centering
\begin{minipage}{0.4\textwidth}
\centering
\resizebox{\linewidth}{!}{
\begin{tabular}{c|c|c|c}
\toprule
ID & Zip Code & Age & Credit Score\\ 
\midrule
1 & 19*** & 60-70 & 797\\
2 & 30*** & 30-40 & 650\\ 
3 & 19*** & 60-70 & 755\\ 
4 & 30*** & 30-40 & 590\\
5 & 19*** & 60-70 & 767\\ 
6 & 30*** & 30-40 & 597\\ 
7 & 30*** & 30-40 & 613\\ 
8 & 19*** & 60-70 & 775\\
\bottomrule 
\end{tabular}}
\end{minipage}
\hspace{10pt}
\begin{minipage}{0.4\textwidth}
\centering
\resizebox{\linewidth}{!}{
\begin{tabular}{c | c | c| c} 
 \toprule
 ID & Zip Code & Age & Credit Score\\ [0.5ex] 
 \midrule
 A & 30*** & 30-40 & 815\\ 
 B & 30*** & 30-40 & 613\\ 
 C & 30*** & 30-40 & 376\\ 
 D & 30*** & 30-40 & 727\\
 \bottomrule
\end{tabular}}
\end{minipage}
\caption{Two $k$-anonymized datasets augmented with credit score information.
{\bf Left:} State Employee Dataset.
{\bf Right:} The Dataset of Company Z.}\label{tab:credit}
\end{table}

\paragraph{Where does $k$-anonymization fail? Reason \#2: Lack of composition.}
A more subtle issue with $k$-anonymity arises from the concept of \emph{composition}, described earlier in this chapter. Unfortunately, $k$-anonymity lacks fundamental composition guarantees and fails when multiple datasets are released. In fact, even releasing just two $k$-anonymized datasets can be sufficient to break its privacy protections in the worst-case scenario. To illustrate this, imagine a situation where it is known that Marcia is a state employee included in a dataset that has been  $4$-anonymized. Additionally, suppose that Marcia is a client of Company $Z$ , which aims at helping individuals improve their credit scores. Company $Z$ sells a separate $4$-anonymized dataset about its customers (see Table~\ref{tab:credit}, Right). Knowing that Marcia is present in both datasets, an adversary can cross-reference the state records with Company $Z$'s records to find a \emph{unique} match: the only individual appearing in both datasets is someone in the 30-40 age range, residing in zip code \texttt{30***}, and having a credit score of 613. This individual must be Marcia, thereby uniquely identifying her. This scenario demonstrates a failure of \emph{composition}: the privacy guarantees of  $k$-anonymity break down when datasets are combined. While this example is simplified for clarity, extensive practical evidence has shown that the issues with  $k$-anonymity are real and pervasive~\citep{narayanan2008robust}. These limitations underscore the need for more robust privacy definitions that can withstand linkage attacks, data aggregation, and the release of multiple datasets.

\subsection{Any Perfectly Accurate and Deterministic Privacy Notion Must Fail} 

Various strategies have been proposed to address the shortcomings $k$-anonymity without significantly reducing the utility of data for demographic and population-level analyses. One such method is \emph{data swapping}, which involves exchanging parts of dataset entries among individuals to ensure that no single row corresponds directly to one person, while still preserving overall demographic counts like the ``number of people in dataset X that have property Y.'' This technique was employed in the release of U.S.~Census data products prior 2020. Another approach is \emph{data minimization}, which focuses on collecting as little data as necessary and discarding it after it has served its purpose.
Despite these efforts, the challenge of ensuring that privacy guarantees degrade gracefully and predictably under repeated queries remains unresolved. To address this, it is important to highlight a \emph{fundamental} property that must be satisfied by any robust privacy definitions, helping us narrow down the search for effective solutions. Specifically, the claim is that \emph{no perfectly accurate and deterministic privacy technique can satisfy our requirements}, and that \emph{randomization is essential for privacy}.

This crucial point can be illustrated with a simple example where the lack of randomness leads to a failure in \emph{composition}. Imagine a hypothetical company named Gluble, which has 25 employees. Gluble publicly announces that the average salary of its employees is \$500,000, perhaps to attract top talent with its competitive compensation. After hiring a 26th employee named Rick, the company updates its public average salary to \$505,000. From these two pieces of information, one can deduce Rick’s salary. Using basic arithmetic, Rick’s salary  $x$ is obtained by solving $\frac{(x + 25 \times 500,000)}{26} = 505,000$, i.e., $x = \$630,000$. 
This amount is significantly higher than his colleagues' salaries. This scenario shows a \emph{failure of composition}: while each individual data release seems innocuous, combining them allows an adversary to infer sensitive information about an individual. Even with access to just two queries, a differential attack reconstructed private data. Although this example is simplified, 
~\cite{dinur2003revealing} have shown that such differential attacks can be executed in far more complex settings, even when the query language is restricted. Importantly, the attack used no information about how the data was privatized. This vulnerability arises because the average salary at Gluble was released deterministically and exactly. What would happen if noise was added to Gluble's salary reports? Suppose that the average salary before hiring Rick was reported as approximately \$500,000, and after hiring, it was approximately \$505,000. It is no longer clear  question whether the change is due to Rick's salary or simply a result of the added randomness. After all, the introduction of noise creates uncertainty, preventing exact inference of individual salaries. This concept of adding randomness to data releases is a cornerstone of Differential Privacy.

Observe that providing Differential Privacy is more complex than ``just'' adding noise. At a high level, the more noise is added, the better our privacy guarantees are going to be; however, adding too much noise is undesirable, as it destroys the utility of privately-released datasets and statistics.  Therefore, noise must be carefully calibrated to balance privacy protection with data utility, enabling us to provide formal and provable privacy guarantees alongside precise \emph{privacy-utility trade-offs}. In fact, three years before Differential Privacy was formally introduced by 
~\cite{dwork2006calibrating}, 
~\cite{dinur2003revealing} laid the groundwork for understanding these trade-offs when incorporating privacy noise. Readers can consult their work, as well as subsequent studies~\citep{DPorg-reconstruction-theory,DPorg-diffix-attack}, for a deeper exploration of the challenges in calibrating noise to protect against reconstruction attacks. The following sections delve deeper into Differential Privacy, what it protects against, its formal definition, guarantees, and the basic mechanisms to achieve it, providing a comprehensive understanding of the crucial role played by randomization in safeguarding privacy.

\subsection{A Side Note: Other Types of Privacy Breaches}

The discussion above highlights the importance of our privacy desiderata and illustrates how previous techniques that failed to meet these criteria have led to significant privacy failures. So far, the presentation relied on simple examples involving variants of anonymization techniques and the challenges associated with privatizing and releasing datasets. However, with the advent of increasingly complex models and large-scale machine learning applications, privacy failures have begun to emerge in more intricate and subtle ways—even when privatized datasets are never directly released. In particular, recent research has demonstrated that privacy can be compromised not only through released statistics but also via the models themselves. A notable example is \emph{Federated Learning (FL)} \citep{li2020federated}. 
The goal of FL frameworks is to protect privacy through decentralization: each user retains their data on their local device, performs computations locally, and only transmits aggregated updates (such as gradient information) to a central server. The intent is that no central entity ever accesses individual user data, thereby preserving privacy. Yet, recent work has shown that this is insufficient: the gradient updates themselves often encode sufficient information to be able to guess the original user data with high accuracy~\citep{zhu2019deep}. 

This issue is not confined to the training of machine learning models. Even after a model is trained and the original data is ostensibly deleted, the released models can still encapsulate information about the training data. This can lead to privacy breaches where models inadvertently memorize and reproduce parts of their training datasets. 
For instance, large language models trained on extensive text corpora have been found to occasionally output verbatim snippets from the training data when prompted in specific ways~\citep{carlini2021extracting}.
A real-world example involves a South Korean AI company Scatter Lab~\citep{theregisterKoreanAppmaker}. Scatter Lab used text and messaging data from users on South Korea's biggest test messaging company, KakaoTalk, to train a chatbot service. Despite efforts to remove personally identifiable information, the chatbot reproduced memorized conversations from the training data when users interacted with it, inadvertently disclosing private and sensitive information about KakaoTalk users. These examples illustrate that privacy breaches can occur even without direct access to the underlying datasets. Thus there is a need for privacy-preserving techniques that extend beyond data anonymization and address the inherent risks in modern machine learning practices. Differential Privacy offers a framework to mitigate these risks by providing formal guarantees that limit the potential for information leakage, even when models are trained on sensitive data and released publicly. Part II of this book explores how Differential Privacy can be applied to machine learning and optimization tasks to safeguard individual privacy for increasingly complex data analysis.


\section{What Protections Does Differential Privacy Provide?}
\label{ch1:sec3}

\subsection{What Does Differential Privacy Promise?}

This section examines and defines what Differential Privacy does and does not protect against. It considers the scenario of an analyst or data curator who aims at collecting and aggregate personal and sensitive data for release in a privacy-preserving manner. This release can take various forms, such as a synthetic version of the dataset that masks private information, a set of sensitive population-level statistics about individuals in the dataset, or a model trained on the sensitive data. The common objective in all these cases is to release data that carefully conceal sensitive attributes at the \emph{individual} and \emph{group} level while retaining sufficient information to provide useful statistics or models at the \emph{population} level. For example, an analyst might wish to determine the fraction of a population with a particular disease or calculate the average salary of employees in a company. In these instances, the data pertaining to each individual is private and sensitive, and individuals may prefer to keep it confidential.

\paragraph{First attempt: No information leakage.} Ideally, no information about any specific individual should be leaked through the data release, i.e., ``nobody can learn \emph{any} information about a specific individual from the privatized computation.'' Achieving this level of privacy is theoretically straightforward: simply do not collect or use any data at all. However, this is impractical, as it precludes any meaningful data analysis. {\em Herein lies a fundamental tension highlighted earlier in this chapter: using more data enhances the accuracy and usefulness of the models and statistics but potentially compromises individual privacy.}

\paragraph{Second attempt: \emph{Almost} no information leakage.} Rather than requiring that one learns \emph{nothing} about any individual when conducting useful statistical analyses, perhaps one can accept learning \emph{as little as possible} or \emph{almost nothing} about them. Recall the example from the introduction concerning Rick’s salary and the addition of noise for privacy. If the company Rick works for is large enough, adding a small amount of noise to the average salary can allow for releasing an approximate estimate of the average salary, while making it difficult to deduce Rick's specific salary. However, the problem here is subtle. It may still be possible to learn significant information about Rick, for instance, that he likely has a high salary because he works at a company where the average salary is close to \$500,000. This may seem innocuous if Rick is expected to hold a high-paying position. But consider a more sensitive scenario: imagine that, in the early 1950s, Rick is a smoker participating in a novel medical study investigating the link between smoking and cancer. The study concludes that smoking does cause lung cancer. As a result, anyone who knows that Rick smokes now knows he is at a higher risk of developing lung cancer. His insurance company might increase his premiums or refuse coverage for cancer treatment, citing a pre-existing condition due to his smoking. Clearly, Rick has been \emph{harmed} by the outcome of the study.

\paragraph{Refining the definition of privacy.}
The perspective adopted in this book and by Differential Privacy is that the above scenario does not constitute a privacy violation. Consider a counterfactual world where Rick did not participate in the study. The medical study would still have concluded that smoking causes cancer, and Rick would have faced the same potential harms. Rick's decision to share his data had \emph{(almost) no impact} on the released statistical inference that smoking causes cancer. This outcome is unavoidable: \emph{any} accurate statistical analysis revealing that smoking causes cancer would have had the same effect on Rick. This book takes the point of view that it is important to distinguish between harms arising from the ethical implications of certain statistical inferences and \emph{privacy harms} that result specifically from the collection and use of an individual's data. This redefines what good statistical privacy guarantees should ensure and the refined desiderata: the goal is to ensure that \emph{one can learn almost nothing new about an individual that could not have been inferred had they not shared their data}. It is important to emphasize that Differential Privacy is not an algorithm; it is a \emph{definition} or \emph{requirement} for privacy. The remainder of this chapter aims at accomplishing two goals: {\it (i)} to carefully formalize the definition and guarantees provided by Differential Privacy, and {\it (ii)} to cover basic algorithmic techniques and building blocks for achieving Differential Privacy.

\begin{figure}[!t]
    \centering
    \includegraphics[width=0.99\linewidth]{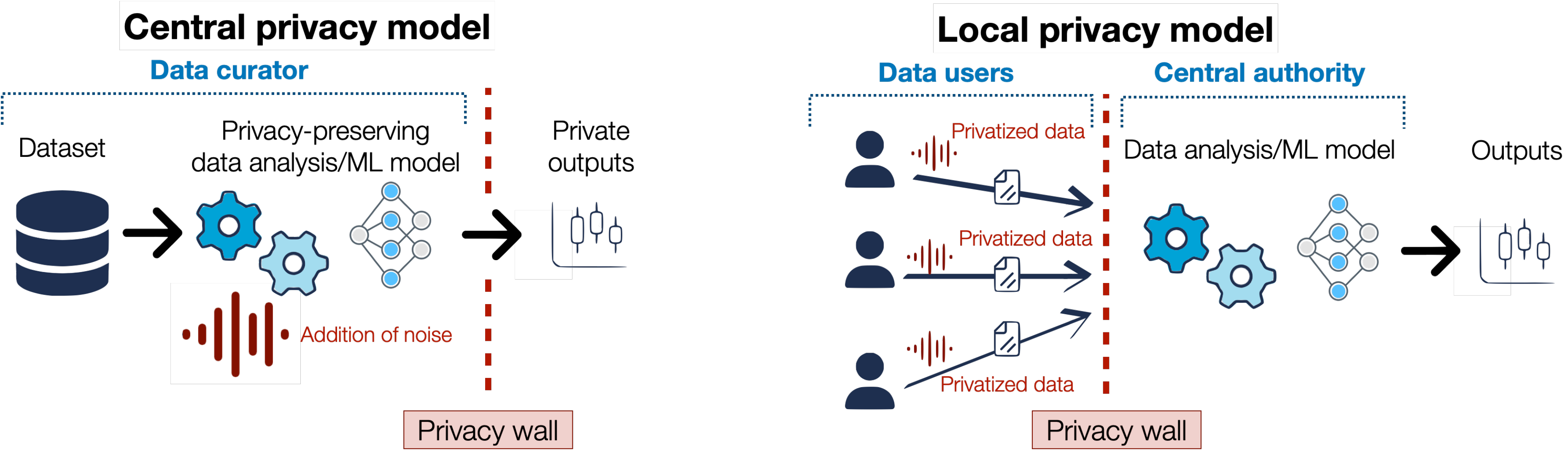}
    \caption{Actors and models in the Privacy Preserving data processing pipeline. Central privacy model (left) and Local privacy model (right).}     \label{fig:privacy_models}
\end{figure}

\subsection{Where to Guarantee Differential Privacy? Local vs Central Models}

Implementing Differential Privacy requires careful consideration of the context in which privacy guarantees are applied, particularly regarding the underlying trust model. The degree of trust placed in data curators or aggregators significantly influences the design and effectiveness of privacy-preserving mechanisms. This book examines the two primary frameworks within privacy-preserving ecosystems: the \emph{centralized} model and the \emph{distributed} (or local) model, each with distinct characteristics and implications for privacy management.

In a \emph{centralized} framework, all data collection, storage, and processing occur at a single, central location managed by a trusted data curator. This central entity has direct access to the raw data and is responsible for implementing and monitoring all privacy-preserving mechanisms. The assumption here is that the data curator will faithfully protect individual privacy and handle data responsibly. This setup represents the \emph{central model} in Differential Privacy, as illustrated in Figure~\ref{fig:privacy_models} (left). Conversely, a \emph{distributed} framework keeps data decentralized, residing at its point of origin—such as personal devices or local databases. Privacy-preserving algorithms are executed locally by the data contributors themselves, and only essential, processed information is communicated to a central authority. For instance, in a typical federated learning setup, raw user data remains on their devices, and only privatized versions of the data—such as noisy data points or gradient updates—are sent to the central aggregator. This approach embodies the \emph{local model} in Differential Privacy, depicted in Figure~\ref{fig:privacy_models} (right).  Both centralized and distributed frameworks offer distinct advantages and challenges concerning privacy. 

Centralized systems concentrate data in one location, creating a single point of failure. If the central entity fails to protect the data—due to a breach or misuse—it can lead to widespread privacy violations affecting all users. Moreover, centralized frameworks require users to \emph{trust} that the platform will implement privacy measures correctly and not exploit the data for unintended purposes. 
However, the centralized setting offers significant advantages in terms of data utility and algorithmic flexibility. Because the data curator has access to the raw data, they can inject carefully calibrated noise at the aggregate level, often requiring much less noise to achieve the same privacy guarantees compared to the local setting. This means that analyses and models derived in the centralized setting can be of higher quality and accuracy. Additionally, the centralized model allows for the development of more complex algorithms that require inspecting the data and estimating joint statistics before adding noise—a process that is often challenging or infeasible in the local model.

In contrast, {\em the distributed model reduces the need for trust in a central authority since privacy is enforced locally by each user}. Even if the central aggregator is compromised, the attacker gains access only to the noisy, privacy-protected data that users have shared. This mitigates the risk associated with a central point of failure and enhances individual control over personal data.
However, while the distributed model enhances privacy by minimizing trust requirements, it also introduces additional complexity in implementing privacy-preserving protocols. Each user must correctly execute the algorithms, which may involve sophisticated computations. Additionally, because each user adds noise to their data independently, the aggregated results may suffer from reduced accuracy due to the accumulation of noise. The centralized model, on the other hand, allows for more efficient privacy-utility trade-offs. Since the data curator has access to the raw data, they can add carefully calibrated noise at the aggregate level, achieving the desired privacy guarantees with potentially less impact on data utility. This centralized addition of noise can result in higher-quality data analyses compared to the distributed approach. 

\paragraph{Distinguishing data privacy from data security.}
It is important to differentiate between \emph{data privacy} and \emph{data security} within the landscape of privacy-preserving technologies. Data \emph{security} focuses on preventing unauthorized access to data, implementing measures such as encryption, authentication protocols, and intrusion detection systems to safeguard against breaches and cyber threats. These measures are designed to protect data from external attackers and unauthorized insiders.  However, security alone is insufficient to prevent the \emph{inference} of individual-level sensitive information from released data. In contrast, data \emph{privacy}, as addressed in this book, aims at preventing \emph{inference} of individual information when data, statistics, or machine learning models are released. Even when cryptographic security is fully implemented,   computing a statistic or training a machine learning model can still allow an attacker to infer individual-level information from the computed statistics or the released model alone, without ever breaching the system or accessing the original data. How this can occur was illustrated through our earlier example of a differential attack recovering Rick's salary or health status. Differential Privacy thus provides an \emph{orthogonal and complementary} layer of protection to traditional data security techniques. {\em While data security aims at preventing unauthorized data access, Differential Privacy limits the potential harm from running inference or reconstruction attacks on released databases, statistics, and models.}

\section{Differential Privacy: Formal Definition, Techniques, and Properties}
\label{ch1:sec4}
Differential Privacy is a mathematical framework for measuring and bounding the individuals' privacy risks in a computation.  The concept, first introduced in 2006 by \citeauthor{dwork2006calibrating} in \cite{dwork2006calibrating}, informally states that the presence or absence of any individual record in a dataset should not significantly affect the outcome of a mechanism.  
In this book, a \emph{mechanism} is defined as any computation that can be performed on the data. Differential Privacy deals with randomized mechanisms, and a mechanism is considered \emph{differentially private} if the probability of any outcome occurring is nearly the same for any two datasets that differ in only one record. 
 
In this context, an adversary is any entity attempting to infer sensitive information about individuals from the output of a data analysis. Remarkably, the privacy guarantee of Differential Privacy holds even if the adversary possesses unlimited computing power and complete knowledge of the algorithm and system used to collect and analyze the data. Thus, even if the adversary were to develop new and sophisticated methods, including the attack methods discussed earlier, as well as new attacks that do not yet exist today, or even if new additional external information becomes available, Differential Privacy  provides the exact same level of protection. In this sense, Differential Privacy is considered \emph{future-proof}. 

The section, next, reviews \emph{Randomized Response}, a classic method adopted in surveys for ensuring the privacy of respondents. Originally developed as a survey technique to encourage honest responses to sensitive questions, Randomized Response leverages randomness to protect individual privacy while still allowing researchers to estimate population characteristics accurately. This method serves as a foundational example of how randomness can be systematically used to achieve Differential Privacy, illustrating the principles that guide more complex privacy-preserving mechanisms discussed later in this section, and throughout the book. 

\paragraph{Randomized response.}
Randomized response~\citep{WarnerRandomResp} was proposed by \citeauthor{WarnerRandomResp} in 1965 to privately survey respondents for a potentially sensitive property. The setup is as follows: one wishes to test for how many individuals in a set of respondents have a certain property, $\cal P$, which might be a controversial one to possess, and this might ordinarily lead to a subset of respondents becoming what Warner described as a ``non-cooperative'' group, who might refuse to be surveyed or provide a dishonest answer, introducing unwanted bias in the survey results. To simultaneously ensure that respondents answer honestly (and, as a result, avoid bias due to the aforementioned non-cooperation) and that their privacy is not violated, randomized response provides respondents the ability to deny their response while also preserving the quality of the summary statistics inferred. This is ensured by introducing randomness into the process of surveying as follows.
\begin{enumerate}[leftmargin=*,topsep=2pt,itemsep=2pt]
    \item The respondent takes a fair coin and flips it;
    \item If \emph{tails} is obtained, then the respondent answers truthfully, and if \emph{heads} is obtained, the respondent flips the coin again and
    \begin{enumerate}[leftmargin=*,topsep=2pt,itemsep=2pt]
        \item Responds affirmatively if the outcome is heads;
        \item Responds negatively if the outcome is tails.
    \end{enumerate}
\end{enumerate}
Note that here the outcomes and numbers of coin flips are only known to the respondent. The property of \textit{plausible deniability} allows respondents to be able to deny their responses, and this provides them with privacy guarantees (as it will be elaborated later). While the responses are partly perturbed due to this process, an analyst can recover the expected number of ``Yes'' responses accurately as follows,
$$
\mathbb{E}[\text{Yes}]=\frac 34 n(\text{has }{\cal P})+\frac 14 n(\text{does not have }{\cal P}),
$$
where $\mathbb{E}[\text{Yes}]$ is the expected number of affirmative responses, and $n(X)$ is the number of respondents who claimed to satisfy property $X$.

As will become clear later in this section, the plausible deniability property of randomized response has a strong connection with Differential Privacy.

\subsection{Differential Privacy, Formally}
\label{ch1:DP_def}

Prior to defining Differential Privacy formally, this section formalizes what this privacy notion aims at protecting (dataset) and the means by which an analyst interacts with data (queries).

\begin{wrapfigure}[10]{r}{250pt}
\vspace{-20pt}
\centering\includegraphics[width=0.9\linewidth]{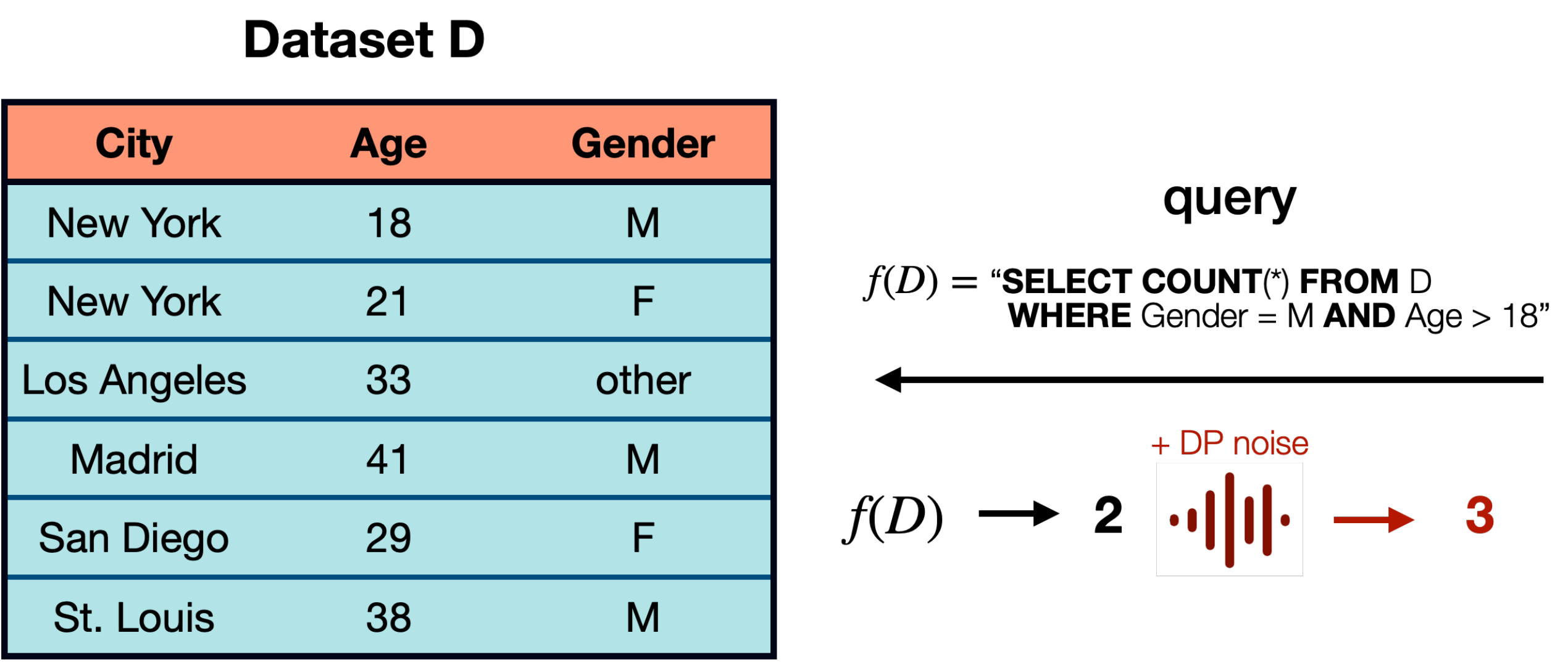}
\caption{Example dataset and query.}
\label{fig:ch1-data}
\end{wrapfigure}

\paragraph{Datasets and queries.}
A \emph{dataset} $D$ is a multi-set of elements in the \emph{data universe} ${\cal U}$. The set of every possible dataset is denoted ${\cal D}$. The data universe ${\cal U}$ is a cross product of multiple \emph{attributes} $U_1, \ldots, U_n$ and has \emph{dimension} $n$. 
For example, Figure \ref{fig:ch1-data}, illustrates a dataset $D$ with three attributes: city, age, and gender. If $C$ is the set of all cities considered, the interval $A=[0,100]$ the set of all ages considered, and $G=\{\text{M}, \text{F}, \text{other}\}$, then $\cal U = C \times A \times G$.
A \emph{numeric query} is a function  $f: {\cal D} \to {\cal R} \subseteq \mathbb{R}^r$ that maps a dataset in some real vector space. For instance, the query $f(D)$ could be an \emph{SQL statement} that counts the number of male individuals over the age of 18 in dataset $D$, as illustrated in Figure \ref{fig:ch1-data}.

The concept of \emph{adjacency} is fundamental in DP. It frames the \emph{unit of change} that Differential Privacy seeks to protect against, ensuring that the presence or absence of any single individual's data does not significantly alter the outcomes of data analysis. There are two common ways to define adjacency in the context of Differential Privacy, reviewed next.

\begin{definition}[Add/remove adjacency]
Two datasets  $D$ and $D'$ are said adjacent under the add/remove notion, denoted as 
$D \sim D'$, if  $|D \Delta D'| = 1$, where $\triangle$ is the symmetric difference of two sets. 
\end{definition}

\noindent
In other words, two datasets are defined as adjacent if one can be obtained from the other by either adding or removing the data of a single individual. This model is particularly relevant when considering the impact of an individual's participation or absence in the dataset. 

\begin{definition}[Exchange adjacency]
Two datasets $D$ and $D'$ are said \emph{adjacent} under the exchange notion, denoted as $D \sim_{\rightleftarrows} D'$, if $D'$ is obtained from $D$ by successively removing one record and then adding a (possibly different) record. That is, there exist elements $d \in D$ and $d' \in \mathcal{U}$ such that:
\(
	D' = (D \setminus \{d\}) \cup \{d'\}.
\)
This implies that $|D| = |D'|$ and $|D \triangle D'| = 2$.
\end{definition}

\noindent
In this notion, adjacent datasets differ in the data of exactly one individual but have the same size. This definition is suited to scenarios where the alteration of data within a constant-size dataset is the primary concern, and can be viewed as the removal followed by the addition of one individual. 

The choice between add/remove or exchange adjacency has some implications for how Differential Privacy is applied as it directly affects the computation of global sensitivity, introduced next, which measures the maximum change in the output of a function for adjacent datasets. This chapter, and generally the book unless specified otherwise, adhere to the add/remove notion of adjacency. 


\paragraph{Global sensitivity.} 
The impact of a single individual's data on the overall analysis is measured through the concept  of \emph{global sensitivity}. Formally, the global sensitivity of a function $f: \mathcal{D} \rightarrow {\cal R}$ is defined as the maximum difference in the output of $f$ over all pairs of adjacent datasets  $D\sim D' \in {\cal D}$, measured with respect to the $\ell_p$ norm: 
\begin{equation}
\label{eq:ch1:sensitivity}
    \Delta_p f = \max_{D \sim D'} \| f(D) - f(D') \|_p.
\end{equation} 
In simpler terms, it measures how much the output of a function can change when an individual's data is added or removed  from the dataset. This measurement provides a basis for determining the amount of noise that needs to be added to the function's output to achieve privacy.
For example, the query considered in Figure \ref{fig:ch1-data} that counts the number of individuals satisfying a certain property in a dataset has global sensitivity $1$, since adding or removing a single individual in the dataset can affect the final count by at most $1$. 
Suppose instead that the task is to compute the average age of all individuals in the dataset. Then the global sensitivity of this average function would be 
\[ 
 \Delta_p f = \frac{\max(A) - \min(A)}{|D|} = \frac{100}{|D|},
\]
where $A$ represents the range of possible ages (assuming ages range from 0 to 100).
In this chapter, the $\ell_1$-sensitivity $\Delta_1 f$ is denoted with $\Delta f$.

\paragraph{Differential Privacy.}
These examples illustrate how a single individual's data can influence the output of a function applied to a dataset. This influence is central to the concept of Differential Privacy. The impact of adding or removing an individual's data varies depending on the type of function in question---whether it's calculating sums, averages, or any other measure of the data. This sensitivity measurement tells us how much  the output of the target function need to be adjusted in order to protect an individual's privacy. Differential Privacy achieves this by adding noise to the function's output, by an amount calibrated to the function sensitivity. This approach ensures that the presence or absence of any single individual's data does not significantly alter the output, thereby masking their participation. 

\begin{definition}[Differential Privacy \citep{dwork2006calibrating}] 
A randomized mechanism ${\cal M} : {\cal D} \to {\cal R}$ with domain ${\cal D}$ and range ${\cal R}$ is \emph{($\varepsilon, \delta$)-differentially private} if, for any event $S \subseteq {\cal R}$ and any pair $D, D' \in {\cal D}$ of adjacent datasets:
\begin{equation} 
\label{eq:dp_def} 
Pr[{\cal M}(D) \in S] \leq \exp(\varepsilon) Pr[{\cal M}(D') \in S] + \delta, 
\end{equation} 
where the probability is calculated over the randomness of ${\cal M}$.
\end{definition}
\noindent A differentially private mechanism maps a dataset to a distribution over the possible outputs because, e.g., it adds random noise or makes randomized choices. The released DP output is a single random sample drawn from this distribution. The level of privacy is controlled by the parameter $\varepsilon \geq 0$, called the \emph{privacy loss}, with values close to $0$ denoting strong privacy, and a secondary parameter $\delta$ which can be loosely interpreted as a margin of error.

First, observe that the inequality: 
\[
Pr[{\cal M}(D) \in S] \leq \exp(\varepsilon) Pr[{\cal M}(D') \in S],
\]
(here with $\delta = 0$ for simplicity of exposition) holds for any $D$ and $D'$. In particular, since it holds for any pair of neighboring databases, it also holds when \emph{swapping} the roles of $D$ and $D'$ in the above definition. Hence, an $(\epsilon,0)$-differentially private algorithm must also satisfy 
\[
Pr[{\cal M}(D') \in S] \leq \exp(\varepsilon) Pr[{\cal M}(D) \in S].
\]
This directly implies the ``stronger'' inequality below:
\begin{align}\label{eq:dp_symmetric}
\exp(-\varepsilon) Pr[{\cal M}(D') \in S] \leq Pr[{\cal M}(D) \in S] \leq \exp(\varepsilon) Pr[{\cal M}(D') \in S],
\end{align}
which highlights that the probabilities of any event $S$ under ${\cal M}$ applied to $D$ and $D'$ are close to each other, controlled by $\varepsilon$.

To intuitively understand these parameters, think of $\varepsilon$ as a knob controlling the level of privacy. Lowering $\varepsilon$ enhances privacy by making the outputs less sensitive to changes in any individual's data. As $\varepsilon$ approaches zero (with $\delta = 0$), the inequality in Equation~\eqref{eq:dp_symmetric} forces the distributions ${\cal M}(D)$ and ${\cal M}(D')$ to become nearly identical. This means \emph{more} privacy, as distinguishing between $D$ and $D'$, which is necessary to recover the data of the individual that differs across both databases, becomes harder. When $\varepsilon = 0$, $Pr[{\cal M}(D) \in S]  = Pr[{\cal M}(D') \in S]$ for all $S$; i.e., the output is independent of and does not use the input dataset, providing \emph{perfect} privacy, but \emph{no utility}---mathematically, an easy implication of $\varepsilon = 0$ is that our mechanism must have a (trivially) constant output across all datasets. When $\varepsilon \to +\infty$, the inequality is always satisfied by any mechanism and \emph{no privacy} is guaranteed.

The parameter $\delta$ serves as a margin of error. It is typically a small number close to zero that defines a \emph{failure threshold} allowing the DP guarantee not to hold with a probability of up to $\delta$\footnote{$(\epsilon,0)$-DP most of the time, except with probability $\delta$, and $(\epsilon,\delta)$-DP are closely related but not exactly equivalent.}. In practice, $\delta$ is chosen to be a negligible value, often much smaller than $\frac{1}{N}$, where $N$ is the size of the dataset. This ensures that the likelihood of disclosing sensitive information about any individual remains extremely low. A mechanism satisfying $(\varepsilon, 0)$-differential privacy is said to satisfy \emph{pure} Differential Privacy or $\varepsilon$-Differential Privacy. 

The guarantees of Differential Privacy are illustrated in Figure~\ref{fig:ch1-data}, which shows the distribution of outputs from a differentially private mechanism applied to two adjacent datasets $D$ and $D'$. The blue and red curves represent the probability distributions of the outputs for $D$ and $D'$, respectively.
The left figure shows how the probability distributions over outputs must be close to each other for adjacent datasets. The right figure quantifies the difference between the probabilities, showing that the log-probabilities of any outcome $x$ differ by at most $\varepsilon$. This means that the ratio of probabilities is bounded by $e^{\varepsilon}$, as required by the definition.
One can see that requiring a smaller $\varepsilon$ forces the distributions to be closer to each other across $D$ and $D'$, making it harder to distinguish between the two databases and hence providing stronger privacy protections. 

\begin{figure}[!t]
    \centering
    \includegraphics[width=0.9\textwidth]{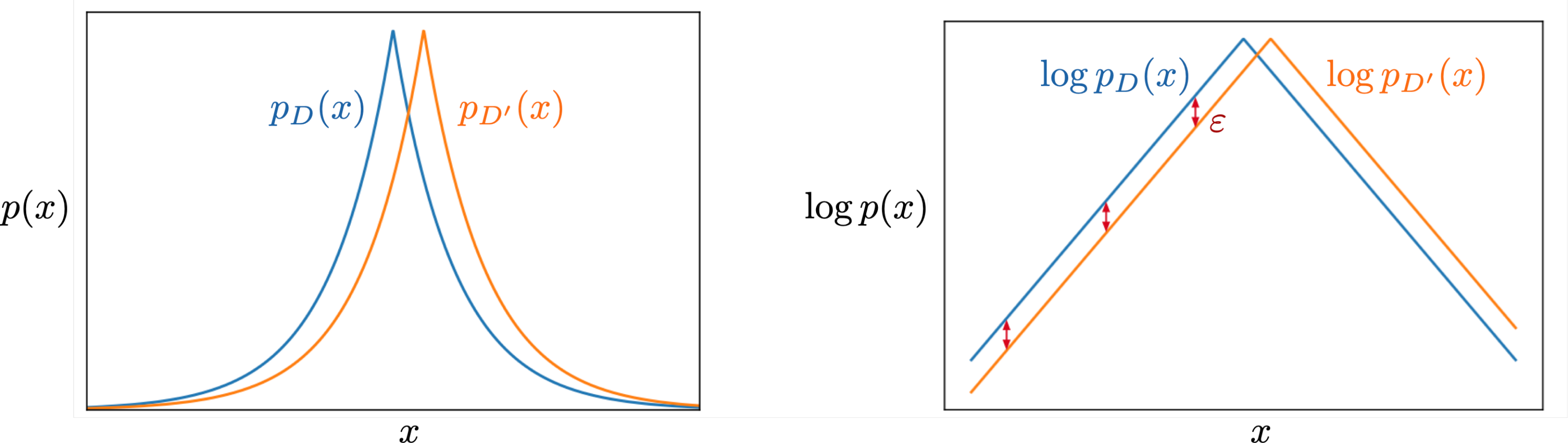}
    \caption{An illustration of the $\varepsilon$-DP guarantee (here, using the Laplace mechanism of Section~\ref{ssub:the_laplace_mechanism}). The log-probability of a value to be output by a mechanism given two neighboring datasets is bounded by $\varepsilon$.}
    \label{fig:laplace_motiv}
\end{figure}


\subsection{Formal Properties of Differential Privacy}
\label{ch1:sec:4.2} 
This section formalizes the properties guaranteed by Differential Privacy, and how they match the desirata described in Section~\ref{ch1:sec2}. The composition, group privacy, and post-processing properties are derived directly from the direction of Differential Privacy, and do not assume a specific mechanism like Randomized Response. As such, composition. group privacy, and post-processing hold for \emph{any} differentially private mechanism, i.e. \emph{any mechanism that satisfies requirement~\eqref{eq:dp_def}}.

\paragraph{Composition.} Composition ensures that a combination of differentially private mechanisms (whether the mechanisms release privatized data, statistics on data, or learning models) preserves Differential Privacy. Composition is a key concept that enables the construction of complex algorithms by combining simpler primitives. It facilitates \emph{privacy accounting}, the rigorous analysis of the overall privacy loss of a composite and potentially complex algorithm by aggregating the privacy guarantees of individual primitives. More formally, it can be stated as follows \citep{DworkRoth2014}:

\begin{theorem}[Composition] \label{th:composabilty} 
    Let ${\cal M}_i:{\cal D} \to {\cal R}_i$ be an $\varepsilon_i$-differentially private 
    mechanism for $i \in \{1, 2\}$. Then, their composition, defined 
    as ${\cal M}(D) = ({\cal M}_1(D), {\cal M}_2(D))$, is $(\varepsilon_1+\varepsilon_2)$-differentially private.  
\end{theorem}

\begin{proof}
For any $(R_1,R_2) \subseteq {\cal R}_1\times {\cal R}_2$ and any two neighboring datasets $D\sim D'$,
    \begin{align*}
        \frac{\Pr[{\cal M}(D) \in (R_1,R_2)]}{\Pr[{\cal M}(D') \in (R_1,R_2)]} &= \frac{\Pr[{\cal M}_1(D) \in R_1]\Pr[{\cal M}_2(D) \in R_2]}{\Pr[{\cal M}_1(D') \in R_1]\Pr[{\cal M}_2(D') \in R_2]}\\
        &= \left(\frac{\Pr[{\cal M}_1(D) \in R_1]}{\Pr[{\cal M}_1(D') \in R_1]}\right)\left(\frac{\Pr[{\cal M}_2(D) \in R_2]}{\Pr[{\cal M}_2(D') \in R_2]}\right)\\
        &\leq \exp(\varepsilon_1)\exp(\varepsilon_2)\\
        &=\exp(\varepsilon_1 + \varepsilon_2).
    \end{align*}
\end{proof}

\noindent
This argument can be generalized to for $k$ differentially private mechanisms by induction. More precisely, if ${\cal M}_i:{\cal D} \to {\cal R}_i$ is an $\varepsilon_i$-differentially private mechanism for $i = 1, \ldots, k$. Then, the composition ${\cal M}(D) = ({\cal M}_1(D), \ldots, {\cal M}_k(D))$ is $(\sum_{i=1}^{k}\varepsilon_i)$-differentially private. The result above is also called \emph{simple composition}, as it deals with pure Differential Privacy mechanisms. 

\paragraph{Group privacy.}
The Differential Privacy notions discussed so far bound differences in output distributions of the mechanism for any pairs of adjacent datasets, i.e. for datasets $D, D'$ such that $\vert D \Delta D'\vert=1$. However, what is not immediately clear is the case when two datasets differ in more than one individual's data. Fortunately, Differential Privacy yields group privacy guarantees that bound this difference for datasets that differ in $k$ entries, for $k>0$:
\begin{theorem}[Group privacy]
Let ${\cal M} : \cal D \to \cal R$ be an $\varepsilon$-differentially private algorithm. Suppose $D$ and $D'$ are two datasets that differ in exactly $k$ entries. Then, for all $S \subseteq \cal R$:
\[
    \Pr[{\cal M}(D) \in S] \leq \exp(k\varepsilon) \Pr[{\cal M}(D') \in S].
\]
\end{theorem}

\begin{proof}
    Let $D^{(0)}\triangleq D$ and $D^{(k)}\triangleq D'$, and let $D^{(0)}, D^{(1)},\ldots,D^{(k-1)},D^{(k)}$ be a sequence of datasets where $D^{(i)}\sim D^{(i+1)}$ for $i=0,1,\ldots,k-1$. The datasets in this sequence can be thought of as ``intermediate'' datasets when trying to obtain $D'$ by starting with $D$ and changing one entry at a time successively. Then by the DP guarantee of ${\cal M}$, for any $R \subseteq {\cal R}$ and $i\in[k-1]$,
    \[\Pr[{\cal M}(D^{(i)}) \in R]\leq\exp(\varepsilon)\Pr[{\cal M}(D^{(i+1)}) \in R].\]

\noindent
    Then, for any $R \subseteq {\cal R}$,
    \begin{align*}
        \Pr[{\cal M}(D) \in R]&=\Pr[{\cal M}(D^{(0)}) \in R]\\
        &\leq\exp(\varepsilon)\Pr[{\cal M}(D^{(1)}) \in R]\\
        &\leq\exp(2\varepsilon)\Pr[{\cal M}(D^{(2)}) \in R]\\
        &\vdots\\
        &\leq\exp(k\varepsilon)\Pr[{\cal M}(D^{(k)}) \in R]\\
        &=\exp(k\varepsilon)\Pr[{\cal M}(D') \in  R].
    \end{align*}
\end{proof}

\paragraph{Post-processing.}
Another key property of Differential Privacy is post-processing immunity. It ensures that privacy guarantees are preserved by arbitrary data-independent post-processing steps \citep{DworkRoth2014}:
\begin{theorem}[Post-Processing Immunity] \label{th:postprocessing} 
    Let $\mathcal{M}: {\cal D} \to {\cal R}$ be a
    mechanism that is $\varepsilon$-differentially private and $g : {\cal R} \to
    {\cal R}'$ be a data-independent mapping. The mechanism $g \circ \mathcal{M}$ is
    $\varepsilon$-differentially private.  
\end{theorem}
\begin{proof}
    The proof first considers a deterministic mapping $g:{\cal R} \to{\cal R}'$. Let $\tilde S\triangleq\{r\in{\cal R}: g(r)\in S\}$, $\forall\,S\subseteq{\cal R}'$. Then for any two neighboring datasets $D\sim D'$,
    \begin{align*}
        \Pr[g\circ{\cal M}(D)\in S]&=\Pr[{\cal M}(D)\in\tilde S]\\
        &\leq\exp(\varepsilon)\Pr[{\cal M}(D')\in\tilde S]\\
        &=\exp(\varepsilon)\Pr[g\circ{\cal M}(D')\in S].
    \end{align*}
    This proves post-processing immunity for deterministic functions. To extend this guarantee to randomized functions, note that randomized functions can be viewed as a distribution over deterministic functions, and, in particular as a convex combination of deterministic functions. Given that a convex combination of differentially private mechanisms (here each mechanism is obtained by composing each deterministic function with the mechanism ${\cal M}$) is also differentially private, the result follows.
\end{proof}

\noindent
This property ensures that, once Differential Privacy guarantees are applied, any further analysis or manipulation of the protected results will not compromise its privacy guarantees. Post-processing  significantly expands the scope and applicability of Differential Privacy algorithms in real-world applications, as shown in Part III.

\paragraph{Quantifiable privacy-accuracy trade-offs.} The last important property, mentioned in Section~\ref{ch1:sec2}, the trade-off between privacy and accuracy can be quantified exactly. Privacy-accuracy trade-offs are \emph{mechanism-level} properties: each mechanism has its own trade-off. The privacy-accuracy trade-offs of the main building blocks are described later in this section, including the privacy-accuracy trade-offs of \emph{Randomized Response} in Section~\ref{ch1:ssub:rr}, of the \emph{Laplace Mechanism} in Section~\ref{ssub:the_laplace_mechanism}, and of the \emph{Gaussian Mechanism} in Section~\ref{ssub:the_gaussian_mechanism}.

\subsection{The Laplace Mechanism}
\label{ssub:the_laplace_mechanism}

\noindent
The Laplace Distribution with 0 mean and scale $b$ has a probability density function $\text{Lap}(x|b) = \frac{1}{2b}
e^{-\frac{|x|}{b}}$. The Laplace mechanism is a differentially private mechanism based on the Laplace distribution for answering numeric queries \citep{dwork2006calibrating}. It is a fundamental building block for many DP algorithms described in this book, and itfunctions by simply computing the output of the query $f$ and then perturbing each coordinate with noise drawn from the Laplace distribution. The scale $b$ of the noise is calibrated to the query sensitivity $\Delta f$ divided by $\varepsilon$:
\begin{definition}[The Laplace Mechanism]
    Let $f: {\cal D} \to {\cal R} \subseteq \mathbb{R}^d$ be a numerical query, with $d$ being a positive integer.  The Laplace
mechanism is defined as ${\cal M}_\text{Lap}(D; f,\varepsilon) = f(D) + Z$ where $Z \in
{\cal R}$ is a vector of i.i.d.~samples drawn from $\text{Lap}(\frac{\Delta f}{\varepsilon})$.
\end{definition}

\noindent
The Laplace mechanism adds random noise drawn from the Laplace distribution independently to each of the $d$ dimensions of the query response.

\begin{theorem}[Differential Privacy of The Laplace Mechanism]
\label{th:m_lap} The Laplace
mechanism, ${\cal M}_\text{Lap}$, achieves $(\varepsilon,0)$-Differential Privacy.
\end{theorem}

\begin{proof}
    Let $D\sim D'$ be any two neighboring datasets in ${\cal D}$, and let $p_D$ and $p_{D'}$ be the probability density functions of ${\cal M}_\text{Lap}(D;f,\varepsilon)$ and ${\cal M}_\text{Lap}(D';f,\varepsilon)$, respectively. 
    Then for any $r\in{\cal R}$,
    \begin{align*}
        \frac{p_D(r)}{p_{D'}(r)} &= \prod_{i=1}^d\left(\frac{\exp\left(-\frac{\varepsilon\vert f(D)_i - r_i\vert}{\Delta f}\right)}{\exp\left(-\frac{\varepsilon\vert f(D')_i - r_i\vert}{\Delta f}\right)}\right)\\
        &= \prod_{i=1}^d\exp\left(\frac{\varepsilon(\vert f(D')_i - r_i\vert-\vert f(D)_i - r_i\vert)}{\Delta f}\right)\\
        &\leq \prod_{i=1}^d\exp\left(\frac{\varepsilon(\vert f(D')_i - f(D)_i\vert)}{\Delta f}\right)&&(\text{By the triangle inequality.})\\
        &= \prod_{i=1}^d\exp\left(\frac{\varepsilon\cdot(\Vert f(D) - f(D')\Vert_1)}{\Delta f}\right)&&(\text{By the definition of $\Delta f$}.)\\
        &\leq\exp(\varepsilon).
    \end{align*}
\noindent
The proof is similar for $\frac{p_{D'}(r)}{p_{D}(r)}\leq\exp(\varepsilon)$.
\end{proof}

\noindent  A graphical representation of the densities and log density of two Laplace distributions associated with neighboring datasets $D$ and $D'$ are provided in Figure \ref{fig:laplace_motiv}, respectively. Note how the difference between the log probabilities for $x$ for each of the neighboring datasets $D\sim D'$ is bounded by $\varepsilon$.

\paragraph{Accuracy guarantee of the Laplace Mechanism.}  The accuracy guarantees of the Laplace Mechanism is characterized by the following result.

\begin{theorem}
For any numerical query $f: {\cal D} \to {\cal R} \subseteq \mathbb{R}^d$, and any database $D \in {\cal D}$,
\[
Pr \left[\left\vert f(D) - {\cal M}_{Lap}(D;f;\varepsilon) \right\vert \geq \ln \left(\frac{d}{\beta}\right) \cdot \left(\frac{\Delta f}{\varepsilon}\right) \right] \leq \beta.
\]
\end{theorem}

\begin{proof}
The proof is for $d = 1$ for simplicity, but it generalizes for $d > 1$. The proof follows from characterizations of  the tails of the Laplace distribution. For a random variable $Z \sim Lap(b)$ and a real number $\alpha > 0$, 
\[
Pr \left[|Z| \geq \alpha \right] = \exp \left(-\alpha/b\right).
\]
Therefore, given that $f(D) - {\cal M}_{Lap}(D;f;\varepsilon)$ is Laplace with parameter $b = \frac{\Delta f}{\varepsilon}$, it follows that
\[
Pr \left[\left\vert f(x) - {\cal M}_{Lap}(D;f;\varepsilon) \right\vert \geq \alpha \right] = \exp \left(-\alpha \cdot \frac{\varepsilon}{\Delta f}\right) \triangleq \beta.
\]
Solving for $\alpha$ in $\exp \left(-\alpha \cdot \frac{\varepsilon}{\Delta f}\right) = \beta$ leads to
\[
\alpha \cdot \frac{\varepsilon}{\Delta f} =  \ln \left(\frac{1}{\beta}\right) ,
\]
hence
\[
\alpha = \ln \left(\frac{1}{\beta}\right) \cdot \left(\frac{\Delta f}{\varepsilon}\right).
\]
This concludes the proof.
\end{proof}

The accuracy guarantee of the Laplace Mechanism provides a practical way to understand how the added noise affects the utility of the released data while ensuring differential privacy. Essentially, it quantifies the expected deviation between the true value of a numerical query and the noisy output produced by the mechanism. 

\subsection{Answering Private Queries in Practice}
Next, we present two examples to illustrate how the Laplace Mechanism can be applied in practice.

\paragraph{Example 1: Computing the average age. }
Consider a dataset containing the ages of $10,000$ individuals, with ages ranging from 0 to 100 years. The task is to compute the average age while ensuring differential privacy. A practical procedure follows the following steps:
\begin{enumerate}[leftmargin=*,topsep=2pt,itemsep=2pt]
\item \textit{Determine the query function and its sensitivity.}
In this task the query function is the average age,
\[
	f(\text{data}) = \frac{1}{n} \sum_{i=1}^n \text{age}_i,
\]
where $n$ is the number of individuals in the dataset.
The global sensitivity $\Delta f$ of the average function is the maximum change in the output when one individual is added or removed. Since the age can vary between 0 and 100, adding the data about a single individual can affect the sum by at most 100 units. Therefore, the sensitivity is:
\[
	\Delta f = \frac{\text{max age} - \text{min age}}{n} = \frac{100 - 0}{10{,}000} = 0.01.
\]

\item \textit{Apply the Laplace Mechanism.}
The next step is to select the privacy parameter $\varepsilon$ and add noise drawn from the Laplace distribution with scale parameter 
$\frac{\Delta f}{\varepsilon}$. Selecting
$\varepsilon = 0.5$ to obtain a strong privacy guarantee adds the following noise:
\[
	\text{noise} \sim \operatorname{Lap}\left( \frac{\Delta f}{\varepsilon} \right) = \operatorname{Lap}\left( \frac{0.01}{0.5} \right) = \operatorname{Lap}(0.02).
\]
The private query thus reports $f(\text{data}) + \text{noise}$.

\item \textit{Analyze the error bound.}
Additionally, by setting a confidence level $\beta = 0.05$ (meaning that one is 95\% confident in the error bound), the error bound can be computed as,
\[
	\text{Error Bound} = \frac{\Delta f}{\varepsilon} \ln\left( \frac{1}{\delta} \right)
 			= \frac{0.01}{0.5} \ln\left( \frac{1}{0.05} \right) 
 			\approx 0.06 \text{ years}.
\]
This means that, with 95\% confidence, the noisy average age returned by the Laplace Mechanism will differ from the true average age by no more than approximately 0.06 years. If  the privacy parameter is set to $\varepsilon = 1$, allowing for slightly less privacy in exchange for greater accuracy, the error bound decreases to about 0.03 years. Thus, selecting $\varepsilon$ and $\beta$ appropriately ensures that the released data remains both useful and privacy-preserving. 
\end{enumerate}

\paragraph{Example 2: Releasing a histogram.}
Suppose a statistical agency wants to release a histogram showing the number of individuals in different age groups, segmented by gender and region, from a dataset containing a large number of respondents. The age groups could be categorized in intervals (e.g., 0--9, 10--19, \dots, 90+). The goal is to release this histogram while ensuring differential privacy. Note that this is different from the previous task where a single quantity wasreleased. The procedure again follows the the three same steps:
\begin{enumerate}[leftmargin=*,topsep=2pt,itemsep=2pt]
\item \textit{Determine the query function and its sensitivity.}
The query function is the count of individuals in each combination of age group, gender, and region. For count queries, the global sensitivity $\Delta f$ is 1 because adding or removing one individual can change the count in one category by at most 1. 

\item \textit{Apply the Laplace Mechanism.}
The next step consists in selecting a privacy parameter $\varepsilon = 0.5$ for each count in the histogram and adding independent Laplace noise to each cell (i.e., each combination of age group, gender, and region) in the histogram. 
Let  $c_{i,j,k}$ be the true count for age group $i$, gender $j$, and region $k$, and $\tilde{c}_{i,j,k}$ is the private counterpart to be released. The counts are linked by the following formula:
\[
	\tilde{c}_{i,j,k} = c_{i,j,k} + \text{Noise}_{i,j,k}, \quad \text{where } \text{Noise}_{i,j,k} \sim \operatorname{Laplace}\left( \frac{\Delta f}{\varepsilon} \right) = \operatorname{Laplace}(2).
\]

\item \textit{Post-processing to ensure valid counts.}
Notice that the application of real-valued noise to each count may render the resulting privacy-preserving counterpart negative or non-integers, thus producing invalid ouputs. These issues can be corrected by applying a post-processing step, that set any negative noisy counts to zero and round the noisy counts to the nearest integer. Such post-processing steps do not alter the privacy guarantees of the original release and are commonly applied in deployments \cite{CDMS21}.

\item \textit{Analyze privacy and utility.}
    Each count is $\varepsilon$-differentially private with $\varepsilon = 0.5$. Since each individual's data affects only one count, and the counts are disjoint, the overall privacy guarantee remains $\varepsilon = 0.5$. The added Laplace noise has a mean of zero and a scale of 2 and thus the expected absolute error for each count is 2. For categories with large counts, this noise has a relatively small impact. However, for categories with small counts, especially in less populated age groups or regions, the noise can significantly affect the accuracy. For a further analysis of disparate impacts of Differential Privacy on different subpopulations, we refer the reader to the survey \cite{FHZ:ijcai22}. 
\end{enumerate}
Note that other mechanisms can produce integer counts directly without additional rounding, by using discrete noise mechanisms, such as the Geometric mechanism~\cite{GhoshRS2012} and the discrete Laplace mechanism~\cite{KarwaSlavkovic2012}.

\section{Approximate Differential Privacy}
\label{ch1:sec5}

The discussion in the previous section focused on pure Differential Privacy and the mechanisms and guarantees associated with it. The case where $\delta>0$ for $(\varepsilon,\delta)$-DP constitutes a variant of Differential Privacy known as \emph{Approximate Differential Privacy}. Recall that $\delta\in(0,1)$ is the failure probability of the privacy loss bound in the relaxed variant of pure DP, and is meant to be a cryptographically low quantity---that is, so small it is considered negligible for practical purposes, often much less than $\frac{1}{N}$ where $N$ is the dataset size. This allows practitioners to apply other mechanisms that yield better utility than the Laplace mechanism in exchange for a marginal failure probability. Importantly, approximate Differential Privacy retains the composition, group privacy, and post-processing immunity properties provided by pure Differential Privacy.

\subsection{The Gaussian Mechanism}
\label{ssub:the_gaussian_mechanism}
The canonical mechanism for $(\varepsilon,\delta)$-DP is the Gaussian mechanism \citep{DworkRoth2014}. Where the Laplace mechanism adds noise proportionally to the $\ell_1$ sensitivity of a query $f$, $\Delta f$, the Gaussian mechanism uses the $\ell_2$ sensitivity, denoted by $\Delta_2 f$, and defined as in Equation \eqref{eq:ch1:sensitivity} with $p=2$. 
The $\ell_2$ and $\ell_1$ norms enjoy the following relationship: for a vector $x \in \mathbb{R}^d, \|x\|_2 \leq \|x\|_1 \leq \sqrt{d}\|x\|_2$ . Thus, the $\ell_2$ sensitivity can be up to a factor $\sqrt{d}$ less than the $\ell_1$ sensitivity.  The Gaussian distribution with 0 mean and standard deviation $\sigma$ has the probability density function ${\cal N}(x|\sigma)=\frac1{\sigma\sqrt{2\pi}}\exp\left(-\frac{(x-\mu)^2)}{2\sigma^2}\right)$. 
\begin{definition}[Gaussian Mechanism]
    Let $f:{\cal D}\to{\cal R}$ be a numerical query. The Gaussian mechanism is defined as ${\cal M}_\text{Gauss}(D;f,\varepsilon) = f(D)+z$ where $z\in{\cal R}$ is a vector of i.i.d. samples drawn from ${\cal N}\left(0,\sigma^2 I\right)$ where $\sigma\geq\sqrt{2\ln(\frac{1.25}{\delta})}(\nicefrac{\Delta_2 f}{\varepsilon})$.
\end{definition}
\noindent
As with the Laplace mechanism, the numerical query response is $d$-dimensional for some integer $d>0$ as well. Gaussian noise is added to each dimension of the query response independently by the Gaussian mechanism. To highlight a key distinction between the Laplace and Gaussian mechanisms, consider the context of computing the mean of a multivariate dataset, revisited from \citep{kamath2020approximate}. Consider a dataset $D \in \{0,1\}^{n \times d}$ aiming to compute the mean in a privacy-preserving manner, denoted by $f(D) = \frac{1}{n} \sum_{i=1}^{n} D_i$. The maximum discrepancy in $f$ across adjacent datasets is $\frac{1}{n} \bm{1}$, yielding a vector with $\ell_1$ norm of $\frac{d}{n}$ and $\ell_2$ norm of $\sqrt{\nicefrac{d}{n}}$ as the $\ell_1$ and $\ell_2$ sensitivities. The following theorem defines the $(\varepsilon,\delta)$-DP guarantees for the Gaussian mechanism.
\begin{theorem}
    The Gaussian mechanism, ${\cal M}_\text{Gauss}$, achieves $(\varepsilon,\delta)$-Differential Privacy, for $\varepsilon\in(0,1]$ and $\delta\in[0,1]$.
\end{theorem}
\noindent
For the proof of this theorem, see Appendix A of \cite{DworkRoth2014}. Notice that, in the original proposition, also reviewed in \cite{DworkRoth2014}, the mechanism is restricted to use $\varepsilon$  within $(0,1]$. However, it is not uncommon to see values of $\varepsilon>1$ in practice, including in various discussions in this book. This restriction was studied and overcome in \cite{Balle2018ImprovingTG}, which provided a more general \emph{analytical Gaussian mechanism} that holds for $\varepsilon>1$ as well. While the details of the DP guarantee of the analytical Gaussian mechanism are beyond the scope of this text, the mechanism and the associated $(\varepsilon,\delta)$-DP is defined as follows.

\begin{theorem}(Analytical Gaussian Mechanism \citep{Balle2018ImprovingTG}). 
    Let $f:{\cal D}\to{\cal R}$ be a numerical query with global $\ell_2$ sensitivity $\Delta_2 f$. $\forall\,\varepsilon>0$ and $\delta\in[0,1]$, the Gaussian mechanism ${\cal M}_\text{Gauss}(D;f,\varepsilon) = f(D)+z$ with $z\sim{\cal N}(0,\sigma^2 I)$ satisfies $(\varepsilon,\delta)$-Differential Privacy if and only if
    \[\Phi\left(\frac{\Delta_2 f}{2\sigma}-\frac{\varepsilon\sigma}{\Delta_2 f}\right)-e^\varepsilon\Phi\left(-\frac{\Delta_2 f}{2\sigma}-\frac{\varepsilon\sigma}{\Delta_2 f}\right)\leq \delta.\]
    Where $\Phi(t)=\Pr[{\cal N}(0,1)\leq t]=\frac{1}{\sqrt{2\pi}}\int_{-\infty}^{t}e^{-\nicefrac{y^2}{2}}dy$ is the CDF of the standard univariate Gaussian distribution.
\end{theorem}
\noindent
The reader is referred to \cite{Balle2018ImprovingTG} for the details on the analytical Gaussian mechanism.


\paragraph{Discussion of accuracy.} The exact formal accuracy guarantees is left as an exercise to the reader. The proof is similar to that of the accuracy guarantee for the Laplace mechanism, simply quantifying tails on the Gaussian distribution. Note that, in high-dimensions, the Laplace mechanism introduces noise scaled by $\frac{d}{n\epsilon}$ to each dimension, providing an $\epsilon$-DP estimate of $f$ with an $\ell_2$ error scaling as $O(\frac{d^{3/2}}{n\epsilon})$. In contrast, the Gaussian mechanism introduces noise with a scale of $O(\sqrt{\frac{d\log(1/\delta)}{n\epsilon}})$ per dimension, resulting in an $(\epsilon, \delta)$-DP estimate of $f$ with $\ell_2$ error approximately $O(\frac{d}{n\epsilon})$. Thus the Gaussian mechanism shaves of a factor of $O(\sqrt{d})$ from the noise, improving accuracy significantly for large $d$ at a slight cost to the privacy guarantee, positing it as a potentially more effective approach for multi-variate estimations.

\section{Beyond Statistical Queries: Differentially Private Selection}
\label{ch1:sec6}

Numerical queries form an important class of computations over which privacy can be enforced. However, in many natural situations, {\em the goal may be to output an object selected according to certain criteria among other objects, rather than just a numerical value.} Consider the following example, adapted from \cite{DworkRoth2014}. Suppose that a retailer is selling an amount of items for which there are 3 potential buyers $A$, $B$, and $C$. Each buyer has a maximum price they are willing to pay for the item, known as their \emph{valuation}. The buyers wish to keep their valuations private, to avoid disclosing sensitive information about their purchasing strategies or financial standing. {\em Hence the task of the retailer is to set a sale price to maximize their total revenue without revealing the valuations of the buyers in the process.}

Assume that the valuations of buyers $A, B$ and $C$ are, respectively \$1.00, \$1.01, and \$3.01. Consider the possible pricing options:
\begin{itemize}[leftmargin=*,topsep=2pt,itemsep=0pt]
    \item \textbf{Price at \$1.00}: All three buyers are willing to purchase at this price, thus the total revenue is \$1.00 $\times$ 3 buyers $=$ \textbf{\$3.00}.
    \item \textbf{Price at \$1.01}: Buyers $B$ and $C$ are willing to purchase, thus the total revenue is \$1.01 $\times$ 2 buyers $=$ \textbf{\$2.02}.
    \item \textbf{Price at \$3.01}: Only buyer $C$ is willing to purchase, thus the total revenue is \$3.01 $\times$ 1 buyer $=$ \textbf{\$3.01}.
\end{itemize}
To maximize revenue, the retailer should set the price at \$3.01. However, since the buyers' valuations are private, the seller cannot directly know the optimal price. The seller needs to select a price in a privacy-preserving manner. One naive approach might be for the seller to add random noise to the buyers' valuations to preserve their privacy. Suppose the seller adds noise to buyer $C$'s valuation, and it becomes, say, \$3.02. Based on this noisy valuation, the seller decides to set the price at \$3.02 apiece. However, this approach leads to a problem: at a price of \$3.02, none of the buyers are willing to purchase the item, since their true valuations are all below this price. Consequently, the total revenue would be \textbf{\$0}, which is worse than any of the previous pricing options. This illustrates that simply adding noise to the valuations is not suitable for such a setting. Adding noise to the valuations can lead to suboptimal pricing decisions. Small changes in the valuations (due to noise) can result in significant differences in the optimal price, which may drastically reduce the seller's revenue or eliminate it altogether. This is particularly problematic when the output is an object selection (the optimal price) rather than a simple numerical query.

\subsection{The Exponential Mechanism}
To be able to perform selection privately while also preserving the quality of the selection made, McSherry and Talwar defined the exponential mechanism \cite{McSherryTalwar2007}. Given a set of objects ${\cal H}$, a dataset $D\in {\cal D}$, and a score function $s: {\cal D}\times{\cal H}\to\mathbb{R}$, the exponential mechanism chooses an object $h\in{\cal H}$ that maximizes the score function in a differentially private manner. 
\begin{definition}[Exponential Mechanism]
    The exponential mechanism, denoted by ${\cal M}_{\text{exp}}$, takes as input a dataset $D\in{\cal D}$, a set of objects $\cal H$, and a score function $s: {\cal D}\times{\cal H}\to\mathbb{R}$ and outputs $h\in{\cal H}$ with probability proportional to $\exp\left(\frac{\varepsilon s(D,h)}{2\Delta s}\right)$, where $\Delta s\triangleq\max_{h\in{\cal H}}\max_{D\sim D'}|s(D,h)-s(D',h)|$.
\end{definition}

\noindent
In this pricing example, the seller defines a utility function $u(D, p)$ that calculates the total revenue generated by setting a price $p$, given the buyers' valuations in the dataset $D$. The exponential mechanism then selects a price $p$ with probability \emph{proportional} -- the actual probability needs to be renormalized to sum to $1$ -- to:
\[
	\exp\left( \frac{\varepsilon \cdot u(D, p)}{2 \Delta u} \right),
\]
where $\varepsilon$ is the privacy parameter controlling the level of privacy, and $\Delta u$ is the global sensitivity of the utility function---that is, the maximum change in $u(D, p)$ when a single individual's valuation in $D$ is modified. The seller thus probabilistic-ally chooses a price that is likely to yield high revenue. 
The probability of selecting a particular price is influenced by the total revenue it generates, but is also smoothed to prevent any single buyer's data from having too much impact on the computation. This smoothing out is controlled by $\varepsilon$. When $\varepsilon \to 0$, all prices become equally likely independently of the buyers' valuations $D$ and the revenue $u(D,p)$, leading to perfect privacy. As $\varepsilon$ increases, the mechanism introduces less smoothing out and gives more importance to the revenue $u(D,p)$, providing more utility---by putting more mass on higher revenues---but less privacy. 
This mechanism thus allows the seller to achieve a balance between maximizing revenue and preserving the privacy of the buyers. The exponential mechanism provides Differential Privacy.
\begin{theorem}
The exponential mechanism, ${\cal M}_{\text{exp}}$, achieves $(\varepsilon,0)$-Differential Privacy.
\end{theorem}

\begin{proof}
The proof assumes that ${\cal H}$ is a finite set. For any two neighbouring datasets $D\sim D'$ and some outcome $h\in{\cal H}$,
\begin{align*}
\frac{\Pr[{\cal M}_{\text{exp}}(D)=h]}{\Pr[{\cal M}_{\text{exp}}(D')=h]}&=\frac{\left(\frac{\exp\left(\nicefrac{\varepsilon s(D,h)}{2\Delta s}\right)}{\sum_{h'\in\mathcal{H}}\exp\left(\nicefrac{\varepsilon s(D,h')}{2\Delta s}\right)}\right)}{\left(\frac{\exp\left(\nicefrac{\varepsilon s(D',h)}{2\Delta s}\right)}{\sum_{h'\in\mathcal{H}}\exp\left(\nicefrac{\varepsilon s(D',h')}{2\Delta s}\right)}\right)}\\
&=\exp\left(\frac{\varepsilon (s(D,h)-s(D',h))}{2\Delta s}\right)\frac{\sum_{h'\in\mathcal{H}}\exp\left(\nicefrac{\varepsilon s(D',h')}{2\Delta s}\right)}{\sum_{h'\in\mathcal{H}}\exp\left(\nicefrac{\varepsilon s(D,h')}{2\Delta s}\right)}\\
&\leq\exp\left(\frac\varepsilon 2\right)\exp\left(\frac\varepsilon 2\right)\frac{\sum_{h'\in\mathcal{H}}\exp\left(\nicefrac{\varepsilon s(D,h')}{2\Delta s}\right)}{\sum_{h'\in\mathcal{H}}\exp\left(\nicefrac{\varepsilon s(D,h')}{2\Delta s}\right)}\\
&=\exp(\varepsilon).
\end{align*}
The inequality follows due to the definition of $\Delta s$.
\end{proof}

\paragraph{Accuracy guarantee.} For the exponential mechanism, accuracy is not measured in terms of how close the mechanism is to the optimal hypothesis $h$. Rather, the objective is to guarantee that, with high probability, the output by the mechanism has a high score, as close as possible to optimality.

\begin{theorem}\label{thm: exp_utility}
Let us fix a database $D$, and let $\mathcal{H}_{\text{OPT}} = \{h^* \in \mathcal{H}~\text{s.t.}~s(D,h) = \max_{h} s(D,h)\}$ be the set of elements in $\mathcal{H}$ that achieve the maximum possible utility score. Then, the exponential mechanism guarantees
\[
Pr \left[s(D, \mathcal{M}_{\text{exp}}(D)) \geq \text{OPT} - \frac{2 \Delta s}{\varepsilon} \left(\ln \left(|\mathcal{H}|/\beta\right) \right)\right] \geq 1-\beta.
\]
where $\text{OPT}= \max_h s(D,h)$.
\end{theorem}

\begin{proof}
Take any $c \in \mathbb{R}$. It follows that
\begin{align*}
    \Pr \left[s(D,\mathcal{M}_{\text{exp}}(D) \leq c \right] 
    &= \frac{\sum_{h:~s(D,h) \leq c} \exp \left(\varepsilon s(D,h)/2 \Delta s\right)}{\sum_{r\in \mathcal{H}} \exp(\varepsilon s(D,h)/2\Delta s)}
    \\&\leq \frac{\sum_{r:~s(D,h) \leq c} \exp \left(\varepsilon c/2 \Delta s\right)}{\sum_{r\in \mathcal{H}_{\text{OPT}}} \exp(\varepsilon \text{OPT}/2\Delta s)}
    \\&\leq \frac{|\mathcal{H}| \exp(\varepsilon c/2\Delta s)}{|\mathcal{H}_{\text{OPT}}| \exp(\varepsilon \text{OPT}/2\Delta s)}
    \\&= \frac{|\mathcal{H}|}{|\mathcal{H}_{\text{OPT}}|} \exp\left(\frac{\varepsilon (c - \text{OPT})}{2 \Delta s}\right)
    \\&\leq |\mathcal{H}| \exp\left(\frac{\varepsilon (c - \text{OPT})}{2 \Delta s}\right).
\end{align*}

The result follows by plugging in 
\[
c \triangleq \text{OPT} - \frac{2 \Delta s}{\varepsilon} \ln \left(|\mathcal{H}|/\delta\right).
\]
\end{proof}

\noindent
Practically, this means that, although the mechanism introduces randomness to protect individual privacy (e.g., the buyers' valuations in our example), it still ensures that the selected output (the price) will yield a utility (the revenue) that is close to the best possible. E.g., in our example, the maximum possible revenue was \$3.01 at price \$3.01). 
Moreover, the utility loss due to privacy is limited and can be controlled by adjusting the privacy parameters.

\section{Randomized Response, Revisited}
\label{ch1:ssub:rr}

Before concluding this chapter, it is useful to revisit the concept of randomized response. Consider Figure \ref{fig:ch1:dpmotiv}: its left side presents a pixelated version of the Mona Lisa, where each pixel is represented by either an `M' or a `.' character. By implementing a random process that flips each pixel with a probability of 0.25, the figure on the right emerges as locally perturbed yet retains the overall image, enabling recognition of the iconic Mona Lisa painting. This metaphor demonstrates that, although plausible deniability is afforded for the original value of each pixel, the outcomes of data analysis can still be preserved with considerable accuracy.

\begin{figure}[!t]
\centering
\includegraphics[width=0.8\linewidth]{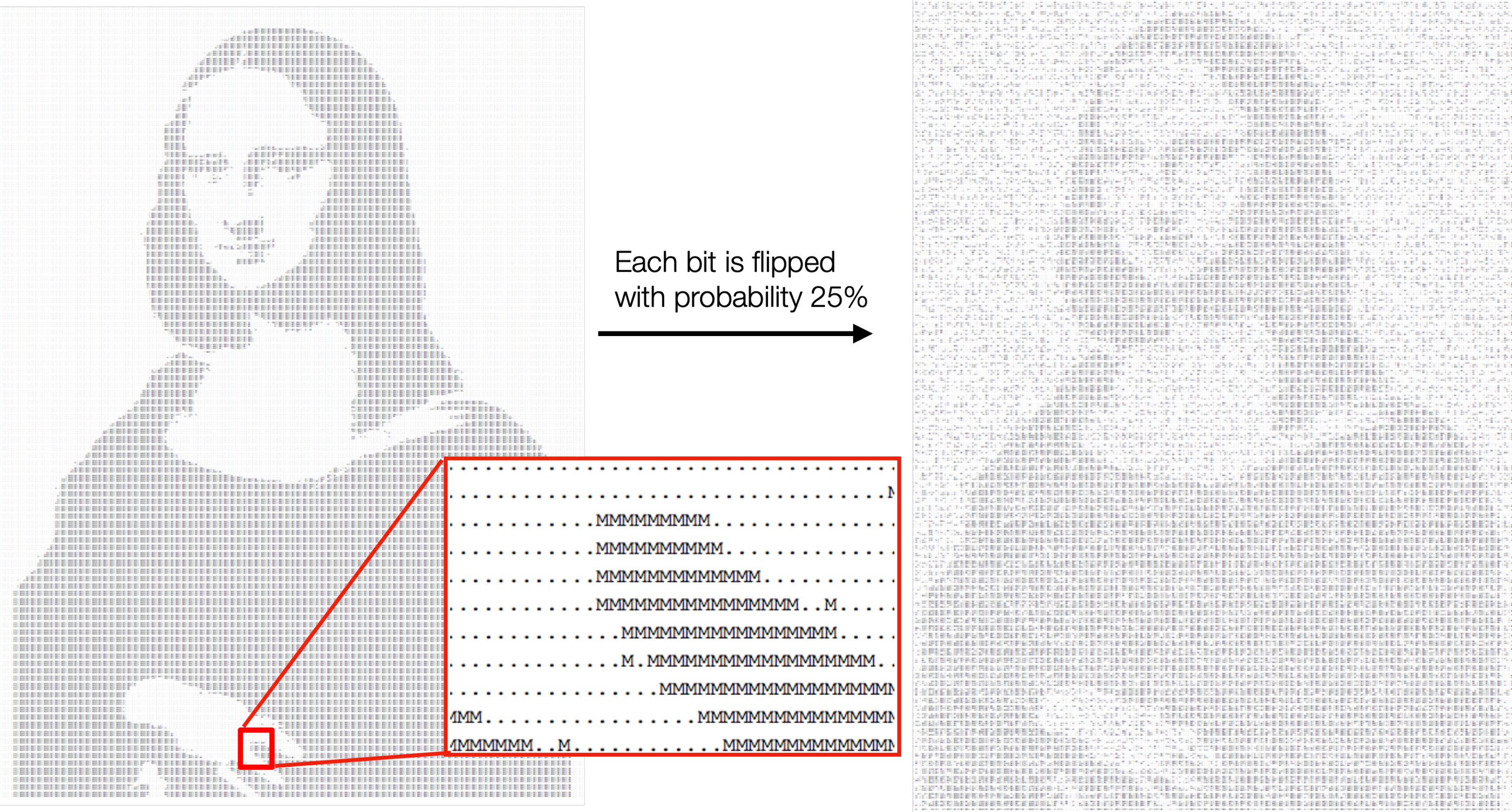}
\caption{A metaphor for private data analysis: Perturbing each bit of the image on the left by flipping it with a random probability of $25\%$ prevents inferring with high probability whether each single bit was originally an "M" or a ".", while still allowing to observe conclusions from the big picture. Figure adapted from slides presentation of Ulfar Erlingsson \cite{RapporTalk2017}.}
\label{fig:ch1:dpmotiv}
\end{figure}

\paragraph{Revisiting Randomized Response.}
Figure \ref{fig:ch1:dpmotiv} happens to be an instance of using randomized response to obscure individual responses while providing accurate summary statistics. Indeed, there is an equivalent formulation of randomized response that satisfies $\varepsilon$-DP in a stronger setting called \emph{local Differential Privacy}, where instead of having a trusted curator that perturbs raw data to provide Differential Privacy, each data contributor perturbs its own data prior to its release. 
Given $\varepsilon>0$, for every private bit $X$, the mechanism is defined as follows:
    $${\cal M}(X)=\begin{cases}X, \text{ with probability}=\frac{\exp{(\varepsilon)}}{1+\exp{(\varepsilon)}};\\
    1-X, \text{ with probability}=\frac{1}{1+\exp(\varepsilon)}.\end{cases}$$

\paragraph{Privacy guarantees.} Randomized response has the following Differential Privacy guarantees.
\begin{theorem}
Randomized Response is $\left(\varepsilon,0\right)$-differentially private.
\end{theorem}

\begin{proof}
Let $p = \frac{\exp(\varepsilon)}{1 + \exp(\varepsilon)}$ for simplicity of exposition. The proof obligation is to upper bound the probability of ratios of probabilities for the two possible outcomes ${\cal M}(X) = X$ and ${\cal M}(X) = 1-X$ for any $X \in \{0,1\}$ and the neighbouring $X' = 1-X$, i.e., 
\[
\frac{Pr [{\cal M}(X) = X]}{Pr [{\cal M}(X') = X]} = \frac{Pr [{\cal M}(X) = X]}{Pr [{\cal M}(1-X) = X]},
\]
and 
\[
\frac{Pr [{\cal M}(X) = 1-X]}{Pr [{\cal M}(X') = 1-X]} = \frac{Pr [{\cal M}(X) = 1-X]}{Pr [{\cal M}(1-X) = 1-X]}.
\]
Note that the first quantity is equal to $\frac{p}{1-p} = \exp(\varepsilon)$, while the second quantity is equal to $\frac{1-p}{p} = \exp(-\varepsilon)$. This is enough to conclude the proof.
\end{proof}

\paragraph{Accuracy of Randomized Response.} To provide the accuracy guarantee of Randomized Response, consider a collection of $n$ data points $X_1, \ldots, X_n$. The goal is to compute the average of these data points, given by $\mu \triangleq \frac{1}{N}\sum_{i=1}^n X_i$. Consider the following simple linear estimator that corrects for the bias introduced by flipping $X$ to the wrong answer, $1-X$, with probability $p \triangleq \frac{\exp(\varepsilon)}{1 + \exp(\varepsilon)}$:
\[
\hat{X} = \frac{1}{(2p-1)N} \left(\sum_{i=1}^n {\cal M}(X_i) + p - 1\right).
\]

\begin{lemma}
$\hat{X}$ is an unbiased estimator of $\mu = \frac{1}{n} \sum_{i = 1}^n X_i$. Further, with probability at least $1 - \beta$, 
\[
\left\vert \hat{X} - \mu \right\vert \leq \frac{\sqrt{1/\beta}}{2 (2p - 1) \sqrt{n}}.
\]
\end{lemma}

\noindent Before providing the proof of this accuracy bound, consider what Differential Privacy promises. Remember that $p \triangleq \frac{\exp(\varepsilon)}{1 + \exp(\varepsilon)}$. Plugging this in the bound above, 
\[
\left\vert \hat{X} - \mu \right\vert = O \left(\frac{(1 + e^\varepsilon)}{2 \left( e^\varepsilon - 1 \right) \sqrt{n}}\right).
\]
As $\varepsilon \to 0$, the $1 + \exp(\varepsilon)$ term goes to $1$; the $1 - \exp(\varepsilon)$ term can be approximated by $\varepsilon$ given a first-order Taylor expansion. Hence, it follows that, as $\varepsilon$ is small,
\[
\left\vert \hat{X} - \mu \right\vert = O\left(\frac{1}{\varepsilon \sqrt{n}}\right).
\]
In particular, given a small $\varepsilon$, to obtain an accuracy of $\alpha$, requires that $n \sim \frac{1}{\varepsilon^2 \alpha^2}$ samples. 

\begin{proof}
Note that 
\begin{align*}
\mathbb{E} \left[ {\cal M}(X) \right] 
&= Pr [{\cal M}(X) = X] \cdot X + Pr [{\cal M}(X) = 1 - X] \cdot (1-X) 
\\&= p X + (1-p) (1-X)
\\&= (2p - 1) X + (1 - p).
\end{align*}
Therefore, 
\[
\mathbb{E} \left[ {\cal M}(X) \right] = (2p - 1) \mu + (1 - p),
\]
immediately implying unbiasedness of $\hat{X}$. Now note that the variance of estimator $\hat{X}$ is given by
\[
\text{Var} \left[\hat{X} \right] 
= \frac{1}{(2p -1)^2 N^2} \sum_{i=1}^N \text{Var}\left[{\cal M}(X_i)\right] 
\leq \sum_{i=1}^N \frac{1}{4 (2p-1)^2 N^2}
= \frac{1}{4 (2p-1)^2 N},
\]
where the first equality follows from the fact that $\text{Var} [cX] = c^2 \text{Var}[X]$ and $\text{Var}[X + c] = \text{Var}[X]$ for a constant $c$, and the inequality follows from the fact that ${\cal M}(X)$ is a Bernoulli random variable and has variance at most $1/4$. Using Chebyshev's inequality with $k = \frac{1}{\sqrt{\beta}}$, it follows that 
\[
\Pr \left[\left\vert \hat{X} - \mu \right\vert \geq \frac{\sqrt{1/\beta}}{2 (1-2p) \sqrt{n}} \right] \leq \beta.
\]
\end{proof}

\noindent
The above bound is an example of privacy-accuracy trade-off. To obtain an accuracy level of $\alpha$ (i.e., the estimator does not mis-estimate $\mu$ by more than $\alpha$) with high probability $1 - \beta$, one needs to pick the value of $p$ such that 
\[
\frac{\sqrt{1/\beta}}{2 (1-2p) \sqrt{n}} \leq \alpha.
\]
This immediately gives the desired value of $\varepsilon$, given us a trade-off between the accuracy level $\alpha$ and the privacy level $\varepsilon$. Here, decreasing $\varepsilon$ towards $0$ (or equivalently decreasing $p$ towards $1/2$) yields a worse accuracy guarantee, as the denominator decreases and eventually goes to $0$. This goes in the expected direction: the more privacy is required, the more the accuracy suffers. 
\section{Concluding Remarks}
\label{ch1:sec7}
This chapter discussed foundational concepts and mechanisms that are the bedrock of Differential Privacy.  Since its conceptual introduction, Differential Privacy has seen considerable evolution, both in theoretical development and practical applications. Researchers have refined the mathematical guarantees, offering tighter bounds on privacy leakage and more effective mechanisms for trading utility with privacy. Practically, Differential Privacy has been applied across diverse sectors, from healthcare to social science, to engineering systems, as reviewed in Part III. These applications demonstrate the flexibility and robustness of Differential Privacy in safeguarding personal information while maintaining data utility. The implications of adopting Differential Privacy extends beyond the technical realm, influencing regulatory policies around data privacy \cite{AIdevelopment2023}, as also discussed in Part V of this book. As organizations increasingly rely on data-driven decision-making, the implementation of DP can help build trust with stakeholders by demonstrating a commitment to privacy-preserving practices. This trust is crucial for compliance with international data protection regulations and for fostering a more privacy-conscious data ecosystem. Furthermore, the principles of Differential Privacy can guide ethical considerations in data usage, promoting a balance between innovation and individual rights to privacy.

\paragraph{Acknowledgements.} 
This work was partially supported by NSF grants SaTC-2345483, CAREER RI-2401285, CAREER HCC-2336236, and by a Google Scholar Research Award. Its view and conclusions are those of the authors only.

\bibliographystyle{unsrtnat}
\bibliography{ref-ch1} 
\end{document}